\documentclass[a4paper,11pt,reqno]{amsart}
\usepackage[body={15.5cm,20.1cm}]{geometry}
\usepackage[T1]{fontenc}
\usepackage{amssymb}
\usepackage[initials]{amsrefs}
\usepackage{mathtools}
\usepackage{accents}
\usepackage{microtype}
\usepackage{enumerate}
\usepackage{graphicx}
\numberwithin{equation}{section}
\theoremstyle{plain}
\newtheorem{thm}{Theorem}[section]
\newtheorem{prop}{Proposition}[section]
\theoremstyle{remark}
\newtheorem{rem}{Remark}[section]
\newtheorem{rems}[rem]{Remarks}
\theoremstyle{definition}
\newtheorem*{matrixrhp*}{Matrix RH problem}

\providecommand{\BS}[1]{\boldsymbol{#1}}

\providecommand{\D}{\mathbb}

\newcommand{\dd}{\mathrm{d}}
\newcommand{\ee}{\mathrm{e}}

\newcommand{\ii}{\mathrm{i}}
\providecommand{\abs}[1]{\lvert#1\rvert}
\providecommand{\accol}[1]{\lbrace#1\rbrace}
\providecommand{\croch}[1]{\lbrack#1\rbrack}
\renewcommand{\Im}{\operatorname{Im}}
\renewcommand{\Re}{\operatorname{Re}}
\DeclareMathOperator{\Ai}{Ai}
\DeclareMathOperator{\diag}{diag}

\DeclareMathOperator{\ord}{O}
\DeclareMathOperator{\osmall}{o}

\DeclareMathOperator{\one}{I}
\DeclareMathOperator{\Res}{Res}

\DeclareRobustCommand{\vcheck}{\accentset{\,\vee}}
\begin{document}
\title[Ostrovsky--Vakhnenko equation by a Riemann--Hilbert approach]{The Ostrovsky--Vakhnenko equation\\by a Riemann--Hilbert approach}
\author[A.~Boutet de Monvel]{Anne Boutet de Monvel$^{\ast}$}
\author[D.~Shepelsky]{Dmitry Shepelsky$^{\dagger}$}
\address{$^{\ast}$%
Institut de Math\'ematiques de Jussieu-PRG,
Universit\'e Paris Diderot,
case 7012, b\^at. Sophie Germain, 
75205 Paris Cedex 13, France}
\email{aboutet@math.jussieu.fr}
\address{$^{\dagger}$%
Mathematical Division,
Institute for Low Temperature Physics,
47 Lenin Avenue,
61103 Kharkiv,
Ukraine}
\email{shepelsky@yahoo.com}
\date{}
\begin{abstract}
We present an inverse scattering transform approach for the equation 
\[
u_{txx}-3 u_x+3u_xu_{xx}+uu_{xxx}=0.
\]
This equation can be viewed as the short wave model for the Degasperis\textendash Procesi equation or the differentiated Ostrovsky--Vakhnenko equation. The approach is based on an associated Riemann\textendash Hilbert problem, which allows us to give a representation for the classical (smooth) solution, to get the principal term of its long time asymptotics, and also to describe loop soliton solutions.  
\end{abstract}
\maketitle
\section{Introduction}                   \label{sec:intro}

We consider the partial differential equation
\begin{equation}                             \label{sDP-om}
u_{txx}-3\kappa u_x+3u_xu_{xx}+uu_{xxx}=0,
\end{equation}
where $\kappa>0$ is a parameter and $u\equiv u(x,t)$ is real-valued. This equation stems from the short-wave limit of the Degasperis--Procesi (DP) equation \cite{DP99}, which is a 
model describing the unidirectional propagation of nonlinear \emph{shallow} water waves:
\begin{equation}                 \label{DP-om}
u_t-u_{txx}+3\kappa u_x+4uu_x=3u_xu_{xx}+uu_{xxx}.
\end{equation}
Indeed, introducing new space and time variables $(x',t')$ and a scaling of $u$ by
\[
x'=\frac{x}{\varepsilon},\qquad t'=\varepsilon t,\qquad u'=\frac{u}{\varepsilon^2}
\]
where $\varepsilon$ is a small positive parameter, then \eqref{sDP-om} is the leading term of \eqref{DP-om} as $\varepsilon\to 0$. Thus, the equation \eqref{sDP-om} can be named as ``the short wave model for the Degasperis\textendash Procesi equation''.

Interestingly, equation \eqref{sDP-om} arises also in the theory of propagation of surface waves in \emph{deep} water, see \cite{KLM11}, as an asymptotic model for small-aspect-ratio waves.

For $\kappa=0$, \eqref{sDP-om} reduces to the derivative Burgers equation
\[
(u_t+u u_x)_{xx}=0,
\]
whereas for $\kappa=-1/3$, \eqref{sDP-om} reduces to the (differentiated) Vakhnenko equation \cites{P93,V92}
\begin{equation}                 \label{Va}
(u_t+ u u_x)_x + u = 0.
\end{equation}
Alternatively, \eqref{sDP-om} with $\kappa=1/3$ reduces to \eqref{Va} after the change of variables $(u,t) \mapsto (-u,-t)$. 

Equation \eqref{Va} arises (and is known as the ``Vakhnenko equation'') in the context of propagation of high-frequency waves in a relaxing medium \cites{V92,V97}. On the other hand, being written in the form
\begin{equation}                 \label{RO}
(u_t+ c_0 u_x + \alpha u u_x)_x  = \gamma u,
\end{equation}
it is also called the ``reduced Ostrovsky equation'' \cite{S06}: it corresponds, in the case $\beta=0$, to the equation
\begin{equation}                 \label{O}
(u_t+ c_0u_x+\alpha uu_x+\beta u_{xxx})_x=\gamma u
\end{equation}
that was derived by Ostrovsky in 1978 \cite{O78}, in the study of weakly nonlinear surface and internal waves in a rotating ocean influenced by Earth rotation. Therefore, it is more correct to name \eqref{Va} the ``Ostrovsky--Vakhnenko (OV) equation'', as it is proposed in \cite{BS13}.

Equation \eqref{O} is also known as the ``Rotation-Modified KdV equation'' (RMKdV); see, e.g., \cites{B05,BC02} and the literature cited there. The term $\beta u_{xxx}$ reflects a small-scale dispersion whereas the term $\gamma u$ (kept in the reduced form \eqref{RO} of the equation) is responsible for a large-scale dispersion due to the influence of Earth rotation (the Coriolis dispersion). In \cite{H90}, Hunter noted that equation \eqref{RO} 
in the form 
\begin{equation}                 \label{H}
(u_t+ u u_x)_x - u = 0
\end{equation}
(which corresponds to $\gamma=-1$ in \eqref{RO}) arises as a short-wave limit of the RMKdV equation and, more generally, it is the canonical asymptotic equation for genuinely nonlinear waves that are non-dispersive as their wavelength tends to zero. This justifies the terminology ``Ostrovsky--Hunter equation'', which is also used for \eqref{H} (see, e.g., \cite{B05}). Again notice that the change of variables $u\to -u, t\to -t$ transforms \eqref{H} to \eqref{Va}.

Equation \eqref{Va} can be reduced to the ``Bullough--Dodd--Mikhailov equation'' \cites{DB77,M79}, see \cites{KLM11, GHJ12}. The transformation of \eqref{Va} to the ``Caudry--Dodd--Gibbon--Sawaga--Kotega equation'' \cites{SK74,CDG76,FO82} is presented in \cite{BS13}.

Well-posedness of the Cauchy problem for the Ostrovsky equation and its relatives (reduced Ostrovsky equation, generalized Ostrovsky equation, etc.) in Sobolev spaces has been widely studied in the literature, using PDE techniques; see \cites{VL04, LM06, SSK10, D13, KM11}. 

On the other hand, equation \eqref{sDP-om} is (at least, formally) integrable: it possesses a Lax pair representation
\begin{subequations}    \label{Lax-ini}
\begin{align}     \label{Lax-x}
\psi_{xxx}&=\lambda(-u_{xx}+\kappa)\psi,\\
\label{Lax-t}
\psi_t&=\frac{1}{\lambda}\psi_{xx}-u \psi_x+u_x \psi,
\end{align}
\end{subequations}
where $\psi\equiv\psi(x,t,\lambda)$. 

In \cite{VP98} the authors introduced a change of variables
\begin{equation}                 \label{u-W}
u(x,t)=U(X,T)=W_X(X,T),\quad x=x_0+T+W(X,T),\quad t=X,
\end{equation}
which reduces \eqref{Va}
to the so-called ``transformed Vakhnenko equation''
\begin{equation}                 \label{W}
W_{XXT}+(1+W_T)W_X = 0.
\end{equation}
These variables turned out to be convenient for applying Hirota's method for constructing 
exact soliton solutions to \eqref{Va} \cites{VP98, MPV99, W10}. These solutions are multi-valued functions having the form of a loop ($1$-soliton) or many loops (multi-solitons).

Another approach to deriving formulas for multi-loop solutions of \eqref{sDP-om} was proposed in \cite{M06}, where these solutions were obtained by taking a scaling limit in the Hirota-type formulas for the  multi-soliton solution of the Degasperis--Procesi equation.

In \cite{VP02} a formalism of the inverse scattering method has been applied to \eqref{W},
which allowed 
to show that the loop solitons can be associated, in a standard way, with the eigenvalues of the 
$X$-equation of the Lax pair associated with \eqref{W}, assuming that the spectral functions
associated with the continuous spectrum do not contribute to the solution
of the inverse spectral problem for the $X$-equation (zero reflection) and thus
the inverse problem can be solved explicitly by linear algebra. The $X$-equation
has the form of a $3\times 3$ matrix ODE 
\[
\Psi_X=(A(\zeta)+B(\zeta,X))\Psi.
\]

Recently \cite{VP12}, this idea has been applied to the case with singular (point) contribution from the continuous spectrum, which allowed to construct particular singular  periodic solutions and solutions combining a singular periodic wave and a soliton. 

In this respect, we notice that the new variable $X$ is in fact the original time variable $t$, and thus the formalism of the inverse scattering method is applied in \cites{VP02,VP12} in such a way that the $t$-equation (in the original variables) is used as spectral problem whereas the $x$-equation provides the ``evolution'' of the spectral data in $x$.

Now we notice that the change of variables \eqref{u-W} is actually the same as we use in the present paper, see Section 2: in our notations, $T=y$, $X=t$, and $W=N$. But we are working with these variables in a different (in a sense, opposite) way. Namely, we keep $t$ as the evolution variable whereas we use $y$ as a space-type variable. This approach allows us to treat the initial value problem for \eqref{sDP-om} (or, equivalently, \eqref{Va}) in a general setting consistent with the natural physical sense of variables: the initial data is a function of $x$, and we are interested in their evolution in $t$
(the natural time variable). 

In this paper we present a Riemann--Hilbert (RH) approach to equation \eqref{sDP-om}, which is based directly on the Lax pair \eqref{Lax-ini} in the form of a pair of $3\times 3$ matrix ODEs. In Section~\ref{sec:rh.formalism} we develop a RH formalism for the OV equation. Particularly, we apply this formalism to the study of the Cauchy problem on the line $x\in(-\infty,\infty)$, assuming that the initial data $u(x,0)=u_0(x)$ are smooth, decay sufficiently fast as $\abs{x}\to\infty$, and satisfy $-u_{0xx}+\kappa >0$ for all $x$. In Section~\ref{sec:loop} we show how the loop solitons can be retrieved in the framework of our Riemann--Hilbert approach. The long time asymptotics of the solution of the Cauchy problem \eqref{icc} is discussed in Section~\ref{sec:as}.

\section{Riemann--Hilbert formalism}   \label{sec:rh.formalism}

Without loss of generality, in what follows we assume that $\kappa=1$. We consider the Cauchy problem
\begin{subequations}   \label{icc} 
\begin{alignat}{2}          \label{ic}
&u_{txx}-3 u_x+3u_xu_{xx}+uu_{xxx}=0,&\qquad&x\in(-\infty,\infty),\ t>0, \\ 
\label{ini-cond}
& u(x,0)=u_0(x),&&x\in(-\infty,\infty),
\end{alignat}
\end{subequations}
where 
\begin{enumerate}[(i)]
\item
$u_0(x)$ is smooth,
\item
$u_0(x)$ decays sufficiently fast as $|x|\to\infty$,
\item
$-u_{0xx}(x)+1 >0$ for all $x$.
\end{enumerate}
In analogy with the DP equation and the Camassa--Holm (CH) equation, one can show that $-u_{xx}(x,t)+1 >0$ for all $(x,t)$ \cite{C01}. The role of this condition in the integrability of equation (\ref{ic}) is discussed in \cite{GHJ12}.

As we have mentioned in Introduction, the problem of existence and uniqueness of solutions for this Cauchy problem can be successfully studied using PDE techniques. On the other hand, it is the inverse scattering methods that show their high efficiency in studying important properties of the solution --- first of all, its long-time behavior \cites{DZ93,BmKST09,BmS08b,BSZ11}. Following this general approach, we propose an inverse scattering formalism, where the solution is represented in terms of the solution of an associated Riemann--Hilbert problem in the complex plane of the spectral parameter. In the case of the DP equation, a similar approach was developed in \cite{BmS13}. 

Similarly to the case of the Degasperis--Procesi equation \cite{BmS13}, it is convenient to introduce the inverse scattering formalism for the Lax pair in the form of a system of first order, $3\times 3$ matrix-valued linear equations, which allows a good control on the behavior of dedicated  solutions of this system as functions of the spectral parameter. 

\subsection{Lax pairs}

Let $z$ be the spectral parameter defined by $\lambda=z^3$. The coefficients of the original Lax pair \eqref{Lax-ini} have singularities at $z=\infty$ and also at $z=0$. In order to have a good control on the behavior of eigenfunctions at $z=\infty$ and at $z=0$ we introduce new forms of \eqref{Lax-ini}, the first one appropriate at $z=\infty$, the second one at $z=0$. 

Introduce 
\begin{equation}\label{La}
\Lambda(z) \equiv \begin{pmatrix}
\lambda_1 & 0 & 0 \\
0 & \lambda_2 & 0 \\
0 & 0 & \lambda_3
\end{pmatrix}=z\begin{pmatrix}
\omega & 0 & 0 \\
0 & \omega^2 & 0 \\
0 & 0 & 1
\end{pmatrix} \equiv z \hat \Lambda,
\end{equation}
where $\omega=\ee^{2\pi\ii /3}$, $\lambda_j=z \omega^j$, $j=1,2,3$, and 
$\hat\Lambda=\diag\{\omega,\omega^2,1\}$.

\begin{prop}[$1$st $3\times 3$ Lax pair] \label{prop:lax1}
The sDP equation \eqref{ic} is the compatibility condition of the system of $3\times 3$ linear equations:
\begin{subequations}   \label{Lax1}
\begin{align}\label{Lax1U}
&\tilde\Phi_x - q \Lambda(z)  \tilde\Phi = U\tilde\Phi, \\ 
\label{Lax1V}
&\tilde\Phi_t +\left(u q \Lambda(z) - \Lambda^{-1}(z)\right)\tilde\Phi = V\tilde\Phi,
\end{align}
\end{subequations}
where $\tilde\Phi\equiv\tilde\Phi(x,t,z)$, 
\begin{subequations}   \label{UV}
\begin{align} \label{m1}
q&=(-u_{xx}+1)^{1/3},\\                       
\label{U}
U&=\frac{q_x}{3 q}\begin{pmatrix}
0 & 1-\omega^2 & 1-\omega \\
1-\omega & 0 & 1-\omega^2 \\
1-\omega^2 & 1-\omega & 0
\end{pmatrix}
 \\
\label{V}
V&=-u U 
+
\frac{1}{3z}\left\{
3\left(\frac{1}{q}-1\right)
I +
 \left(q^2-\frac{1}{q}\right)\begin{pmatrix}
1 & 1 & 1 \\
1 & 1 & 1 \\
1 & 1 & 1
\end{pmatrix}
\right\}
\begin{pmatrix}
\omega^2 & 0 & 0 \\
0& \omega & 0 \\
0 & 0 & 1
\end{pmatrix}.
\end{align}
\end{subequations}
Here $I$ is the $3\times 3$ identity matrix.
\end{prop}

\begin{proof}
Let $\Phi\equiv\Phi(x,t,z)$ be the vector-valued function defined by
\begin{equation}\label{Phi}
\Phi
=\begin{pmatrix}\psi\\ \psi_x \\ \psi_{xx}\end{pmatrix}. 
\end{equation}
Then the Lax pair \eqref{Lax-ini} can be written in the matrix form:
\begin{subequations}   \label{Lax1-1}
\begin{align}\label{Lax1-1U}
&\Phi_x = \begin{pmatrix}
0 & 1 & 0\\
0 & 0 & 1 \\
z^3 q^3 & 0 & 0
\end{pmatrix}\Phi,\\
\label{Lax1-1V}
&\Phi_t = \begin{pmatrix}
u_x & -u & z^{-3}\\
1 & 0 & -u \\
-z^3 u q^3 & 1 & -u_x
\end{pmatrix}\Phi.
\end{align}
\end{subequations}
Introduce 
\begin{subequations}   \label{D-P}
\begin{align}\label{D-P-D}
D(x,t) &= \begin{pmatrix}
q^{-1}(x,t) & 0 & 0 \\
0 & 1 & 0 \\
0 & 0 & q(x,t)
\end{pmatrix},\\
\label{D-P-P}
P(z) &= \begin{pmatrix}
1 & 1 & 1\\
\lambda_1 & \lambda_2 & \lambda_3 \\
\lambda_1^2 & \lambda_2^2 & \lambda_3^2
\end{pmatrix}\equiv 
\begin{pmatrix}
1 & 0 & 0 \\
0 & z & 0 \\
0 & 0 & z^2
\end{pmatrix}
\begin{pmatrix}
1 & 1 & 1\\
\omega & \omega^2 & 1 \\
\omega^2 & \omega & 1
\end{pmatrix}.
\end{align} 
\end{subequations}
Setting $\tilde\Phi = P^{-1}D^{-1} \Phi$, the system \eqref{Lax1}-\eqref{UV} for $\tilde\Phi$ follows from \eqref{Lax1-1}.
\end{proof}

Notice that the Lax pair \eqref{Lax1}-\eqref{UV} is appropriate for controlling the behavior of its solutions for large $z$ because $U\equiv U(x,t,z)$ and $V\equiv V(x,t,z)$ are bounded at $z=\infty$ and, moreover, the diagonal part of $U$ vanishes identically while the diagonal part of $V$ is $\ord(1/z)$ as $z\to\infty$ (see \cites{BC,BmS08a,BmS13}). 

\begin{prop}[$2$nd $3\times 3$ Lax pair]\label{prop:lax2}
The sDP equation \eqref{ic} is the compatibility condition of the system of $3\times 3$ linear equations:
\begin{subequations} \label{Lax0}
\begin{align}\label{LaxU0}
\tilde\Phi_{0x}-\Lambda(z)\tilde\Phi_0 &= U_0\tilde\Phi_0, \\ 
\label{LaxV0}
\tilde\Phi_{0t}-\Lambda^{-1}(z)\tilde\Phi_0 &= V_0\tilde\Phi_0,
\end{align}
\end{subequations}
where $\tilde\Phi_0\equiv\tilde\Phi_0(x,t,z)$, 
\begin{subequations}  \label{UV0}
\begin{align}                    
\label{U0}
U_0&=-\frac{z u_{xx}}{3} \begin{pmatrix}
\omega & 0 & 0 \\
0 & \omega^2 & 0 \\
0 & 0 & 1
\end{pmatrix}\begin{pmatrix}
1 & 1 & 1\\
1 & 1 & 1\\
1 & 1 & 1
\end{pmatrix},\\
\label{V0}
V_0&=\frac{u_x}{3}\begin{pmatrix}
0  & 1-\omega^2 & 1-\omega \\
1-\omega & 0 & 1-\omega^2 \\
1-\omega^2 & 1-\omega & 0
\end{pmatrix} - z u \begin{pmatrix}
\omega & 0 & 0 \\
0 & \omega^2 & 0 \\
0 & 0 & 1
\end{pmatrix}\left\{I-\frac{u_{xx}}{3} \begin{pmatrix}
1 & 1 & 1\\
1 & 1 & 1\\
1 & 1 & 1
\end{pmatrix}
\right\}.
\end{align}
\end{subequations}
\end{prop}

\begin{proof}
Introduce  $\tilde\Phi_0\equiv\tilde\Phi_0(x,t,z)$  by $\tilde\Phi_0=P^{-1}\Phi$. Then \eqref{Lax0}-\eqref{UV0} for $\tilde\Phi_0$ follows from \eqref{Lax1-1}.
\end{proof}

Since $U_0\equiv U_0(x,t,z)$ and $V_0\equiv V_0(x,t,z)$ are bounded at $z=0$, the Lax pair \eqref{Lax0}-\eqref{UV0} is appropriate for controlling the behavior of its solutions as $z\to 0$.

Moreover, $U_0(x,t,0)\equiv 0$ for all $(x,t)$. This property will be used below, in establishing the relationship between the solution of the associated Riemann--Hilbert problem and the solution of equation \eqref{ic}.

\subsection{Eigenfunctions}

Now assume that $u(x,t)$ is a solution of the Cauchy problem \eqref{icc} and define dedicated solutions of the systems \eqref{Lax1} and \eqref{Lax0}, taking into account that $U,V,U_0$, and $V_0$ all decay to $0$ as $|x|\to \infty$. 

The l.h.s.\ of \eqref{Lax1} suggest introducing a matrix-valued function $Q(x,t,z)$ satisfying the system of equations
\begin{align*}
Q_x&=q(x,t)\Lambda(z),\\
Q_t&=-u(x,t)q(x,t)\Lambda(z)+\Lambda^{-1}(z).
\end{align*}
Introducing the new variable 
\begin{equation}\label{y}
y(x,t)\coloneqq x-\int_x^\infty\left(q(\xi,t)-1\right)\dd\xi,
\end{equation}
so that $\partial y/\partial x= q$, one defines $Q$ by
\begin{equation}\label{Q}
Q(x,t,z) = y(x,t)\Lambda(z)+t \Lambda^{-1}(z).
\end{equation}
Notice that the equality  $Q_t=-uq\Lambda+\Lambda^{-1}$ follows from the equality 
$q_t=-(uq)_x$, which is an equivalent form of \eqref{ic}.

The role of $Q(x,t,z)$ is to catch the large-$z$ behavior of solutions of the Lax equations \eqref{Lax1}.
Indeed, introducing $M(x,t,z)$ by 
$$
M=\tilde \Phi \ee^{-Q}
$$
reduces \eqref{Lax1} to the system 
\begin{subequations}   \label{M}
\begin{align}\label{MU}
&M_x - [Q_x,M]  = UM, \\ 
\label{MV}
&M_t - [Q_t,M]  = VM,
\end{align}
\end{subequations}
(brackets denote matrix commutator), whose solutions can be constructed as solutions of the Fredholm integral equation
\begin{align} \label{M-int}
&M(x,t,z)=\notag\\
&I + \int_{(x^*,t^*)}^{(x,t)} \ee^{Q(x,t,z)-Q(\xi,\tau,z)}
	\left(  U M(\xi,\tau,z)\dd\xi +  V M(\xi,\tau,z)\dd\tau\right)\ee^{-Q(x,t,z)+Q(\xi,\tau,z)}.
\end{align}
The matrix equation \eqref{M-int} has to be understood as a collection of scalar integral equations, where the initial point of integration $(x^*,t^*)$ can be chosen differently for different matrix entries of the equation, e.g., $(x_{jl}^*,t_{jl}^*)$ for $M_{jl}$. Particularly, for the Cauchy problem considered here, it is reasonable to reduce the integration in \eqref{M-int} to paths parallel to the $x$-axis, i.e., $t_{jl}^*=t$ and to choose $x_{jl}^*=\infty_{jl}=\pm\infty$ as the initial point of integration in such a way that the exponentials in \eqref{M-int} provide analyticity and boundedness for the matrix entries. Namely, defining 
\begin{equation}\label{signs}
\infty_{jl}=
\begin{cases}
+\infty,&\text{if }\Re\lambda_j(z)\geq\Re\lambda_l(z),\\ 
-\infty, &\text{if }\Re\lambda_j(z)<\Re\lambda_l(z),
\end{cases}
\end{equation}
and taking into account that $V\to 0$ as $\abs{x}\to \infty$, equation \eqref{M-int} reduces to the system of Fredholm integral equations
\begin{equation} \label{M-int3}
M_{jl}(x,t,z)=
I_{jl}+\int_{\infty_{j,l}}^x\ee^{-\lambda_j(z)\int_x^\xi q(\zeta,t)\dd\zeta}
\croch{(UM)_{jl}(\xi,t,z)}\ee^{\lambda_l(z)\int_x^\xi q(\zeta,t) \dd\zeta}\dd\xi.
\end{equation}

Introduce the set $\Sigma=\{z\mid\Re\lambda_j(z)=\Re\lambda_l(z)\text{ for some } j\neq l\}$. By \eqref{La}, $\lambda_j=\omega^j$, hence $\Sigma$ consists of six rays
\[
l_{\nu}=\D{R}_+\ee^{\frac{\pi\ii}{3}(\nu-1)},\quad\nu=1,\dots,6
\]
and divides the $z$-plane into six sectors
\[
\Omega_{\nu} = \Bigl\lbrace z\,\Bigm\vert\,\frac{\pi}{3}(\nu-1)<\arg z < \frac{\pi}{3}\nu\Bigr\rbrace,\quad\nu=1,\dots,6.
\]

\begin{prop}[see \cite{BC}]\label{p1}
Let $M(x,t,z)$ be the solution of the system of equations \eqref{M-int3}, 
where the limits of integration are chosen according to \eqref{signs}. Then
\begin{enumerate}[\rm(i)]
\item
$M$ is piecewise meromorphic with respect to $\Sigma$, as function of the spectral parameter $z$. 
\item
$M(x,t,z)\to I$ as $x\to+\infty$ and $M(x,t,k)$ is bounded as $x\to -\infty$ for all $z\in {\mathbb C}\setminus\Sigma$ where $M$ is regular.
\item
$M(x,t,z)\to I$ as $z\to\infty$.
\end{enumerate}
\end{prop}

Equations \eqref{M-int3} provide only the boundedness of $M_{j,l}(\,\cdot\,,\,\cdot\,,z)$ as $z\to 0$ in the corresponding sector $\Omega_{\nu}$.
In order to have better control of the behavior of  solutions as $z\to 0$, 
it is convenient to use the Lax pair \eqref{Lax0}. 

Introducing $M_0(x,t,z)$ by 
$$
M_0=\tilde \Phi_0 \ee^{-Q_0},
$$
where $Q_0=\ee^{x\Lambda+t\Lambda^{-1}}$, 
reduces \eqref{Lax0} to the system 
\begin{subequations}   \label{M0}
\begin{align}\label{MU0}
&M_{0x} - [Q_{0x},M_0]  = U_0 M_0, \\ 
\label{MV0}
&M_{0t} - [Q_{0t},M_0]  = V_0 M_0,
\end{align}
\end{subequations}
whose solutions, in analogy with $M$, can be constructed as solutions of the Fredholm integral equation
\begin{equation} \label{M-int3-0}
M_{0jl}(x,t,z)=
I_{jl}-\frac{z}{3}\int_{\infty_{j,l}}^x\ee^{\lambda_j(z)(x-\xi)}
\croch{(u_{xx}\Omega M_0)_{jl}(\xi,t,z)}\ee^{-\lambda_j(z)(x-\xi)}\dd\xi,
\end{equation}
where $\Omega=\left(\begin{smallmatrix}
\omega & \omega & \omega \\
\omega^2 & \omega^2 & \omega^2 \\
1 & 1 & 1
\end{smallmatrix}\right)$.

Similarly to Proposition~\ref{p1}, equation \eqref{M-int3-0} determines a piecewise meromorphic, $3\times 3$ matrix-valued function $M_0$. Moreover, the particular dependence on $z$ in \eqref{M-int3-0} implies a particular form of the first coefficients in the expansion of $M_0$ as $z\to 0$.

\begin{prop} \label{p0}
Let $M_0(x,t,z)$ be the solution of the system of equations \eqref{M-int3-0}, 
where the limits of integration are chosen according to \eqref{signs}. Then
\begin{enumerate}[\rm(i)]
\item
$M_0$ is piecewise meromorphic with respect to $\Sigma$, as function of the spectral parameter $z$. 
\item
$M_0(x,t,z)\to I$ as $x\to +\infty$ and $M_0(x,t,z)$ is bounded as $x\to -\infty$ for all $z\in {\mathbb C}\setminus\Sigma$ where $M_0$ is regular.
\item
$M_0(x,t,z)\to I$ as $z\to 0$. Moreover,
\begin{align} \label{M0-as}
&M_0(x,t,z) = I + M_0^{(1)}(x,t) z + M_0^{(2)}(x,t) z^2 + \ord(z^3)\  \text{as}\  z\to 0,
\intertext{where} 
\label{M0-coef}
&M_0^{(1)}(x,t) = -\frac{1}{3} u_x \Omega, \qquad M_0^{(2)}(x,t) = -\frac{1}{3} u \tilde\Omega,
\end{align}
with $\tilde\Omega = \hat\Lambda \Omega - \Omega \hat\Lambda$ and $\hat\Lambda = \diag\{\omega, \omega^2, 1\}$.
\end{enumerate}
\end{prop}

\begin{rem}
The important property of \eqref{M0-as} is that its terms, up to the second order, do not depend on the sector of the $z$-plane where the limit is taken.  This is due to the fact that the integrals $\int_{\pm\infty}^x u_{\xi\xi}\dd\xi=-u_x$ and $\int_{\pm\infty}^x (x-\xi) u_{\xi\xi}\dd\xi = -u$, which arise when calculating the expansion, do not depend on the sign at infinity. However, the further terms in \eqref{M0-as} depend, in general, on the sector.
\end{rem}

Now, noticing that $M$ and $M_0$ are related to the same linear system of PDEs \eqref{Lax-ini}, tracing back the way that the differential equations for $M$ and $M_0$ were derived, and taking into account their behavior for large $x$ lead to the following

\begin{prop}\label{prop:rel}
The functions $M$ and $M_0$ are related as follows:
\begin{align}\label{M-relat}
&M(x,t,z) = G(x,t) M_0(x,t,z) \ee^{N(x,t)z\hat\Lambda},
\intertext{where} 
\label{N}
&N(x,t) = x - y(x,t) = \int_x^{+\infty} (q(\xi,t)-1)\dd \xi
\intertext{and}
\label{G}
&G(x,t) = P^{-1}(z) D^{-1}(x,t) P(z)  = \begin{pmatrix}
\alpha & \beta &  \bar\beta   \\
\bar\beta  & \alpha  & \beta     \\
\beta  & \bar\beta  & \alpha
\end{pmatrix}
\intertext{with}
\label{al-be}
&\alpha = \bar{\alpha} = \frac{1}{3}\left(q+1+\frac{1}{q}\right), \quad
\beta = \frac{1}{3}\left(q+\omega^2+\frac{\omega}{q}\right).
\end{align}
\end{prop}

\begin{rem}
In spite of the fact that $P^{-1}(z)=\frac{1}{3}\left(\begin{smallmatrix}
1 & \omega^2 & \omega \\
1 & \omega & \omega^2 \\
1 & 1 & 1 
\end{smallmatrix}\right)\left(\begin{smallmatrix}
1 & 0 & 0 \\
0 & z^{-1} & 0 \\
0 & 0 & z^{-2} 
\end{smallmatrix}\right)$ is singular at $z=0$, the factor $G$ is nonsingular. Moreover, it is independent of $z$.
\end{rem}

From \eqref{M-relat}-\eqref{Q} and \eqref{M0-as}-\eqref{M0-coef} we derive the following expansion of $M$ as $z\to 0$: 
\begin{equation} \label{M-as}
M = G\left(I + z\left\{-\frac{u_x}{3}\Omega + N \hat\Lambda\right\} 
+ z^2 \left\{-\frac{u}{3}\tilde\Omega -\frac{u_x}{3} N \Omega \hat\Lambda
+ \frac{N^2}{2}\hat\Lambda^2\right\}  + \ord(z^3)\right).
\end{equation}

\begin{prop}[symmetries] \label{p-sym}
$M(x,t,z)$ satisfies the symmetry relations:
\begin{enumerate}[\rm{(S}1)]
\item 
$\Gamma_1 \overline{M(x,t,\bar z)} \Gamma_1 = M(x,t,z)$ with $\Gamma_1 = 
\left(\begin{smallmatrix}
0 & 1 & 0 \\
1 & 0 & 0 \\
0 & 0 & 1
\end{smallmatrix}\right)$.
\item
$\Gamma_2 \overline{M(x,t,\bar z\omega^2)} \Gamma_2 = M(x,t,z)$ with $\Gamma_2 = 
\left(\begin{smallmatrix}
	0 & 0 & 1 \\
	0 & 1 & 0 \\
	1 & 0 & 0
\end{smallmatrix}\right)$.
\item 
$\Gamma_3 \overline{M(x,t,\bar z \omega)} \Gamma_3 = M(x,t,z)$ with $\Gamma_3 = 
\left(\begin{smallmatrix}
	1 & 0 & 0 \\
	0 & 0 & 1 \\
	0 & 1 & 0
\end{smallmatrix}\right)$.
\item
$M(x,t,z\omega) = C^{-1}M(x,t,z)C$, with $C=\left(\begin{smallmatrix}
	0 & 0 & 1 \\
	1 & 0 & 0 \\
	0 & 1 & 0
\end{smallmatrix}\right)$.
\end{enumerate}
\end{prop}

Actually, assuming any two symmetries from above, the other two follow.

\subsection{Matrix RH problem}
\subsubsection{Jump conditions}

For $z$ on the common boundary of two adjacent domains $\Omega_{\nu}$, the limiting values of $M$, being the solutions of the system of differential equations \eqref{M}, must be related by a matrix independent of $(x,t)$. 
\begin{figure}[ht]
\centering\includegraphics[scale=1]{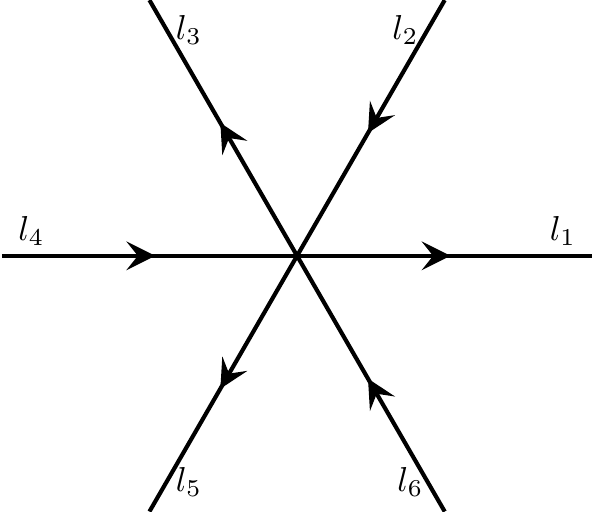}
\caption{Rays $l_{\nu}$ in the $z$-plane} 
\label{fig:sectors}
\end{figure}
Supplying the rays $l_{\nu}$ with an orientation, see Figure~\ref{fig:sectors}, 
we can write for the limiting values of $M$:
\begin{equation}\label{scat}
M_+(x,t,z)=M_-(x,t,z)\ee^{Q(x,t,z)}S_0(z)\ee^{-Q(x,t,z)},\quad z\in l_1\cup\dots\cup l_6.
\end{equation}
Considering \eqref{scat} at $t=0$ we see that $S_0(z)$ is determined by $u(x,0)$, i.e., by the initial data for the Cauchy problem \eqref{icc}, via the solutions $M(x,0,z)$ of the system of integral equations \eqref{M-int3} whose coefficients are determined by $u(x,0)$. Thus the relation \eqref{scat} can be considered as a ``pre-Riemann--Hilbert problem'' associated with \eqref{icc}: the data are $S_0(z)$, and we seek for a piecewise meromorphic function $M$ satisfying \eqref{scat} for all $(x,t)$, in the hope that one can further extract the solution $u(x,t)$ to \eqref{icc} by evaluating $M(x,t,z)$. 

The matrix $S_0(z)$ has a particular matrix structure, see \cite{BmS13}. Indeed, the integral equations \eqref{M-int3} allow studying the limiting values of $M$ as $x\to\pm\infty$ for $z\in\Sigma$. Set $t=0$ and consider, for example, the limiting values $M_\pm(x,0,z)$ for $z\in l_1=\D{R}_+$. For such $z$, $\Re\lambda_1(z)=\Re\lambda_2(z)<\Re\lambda_3(z)$, and thus the structure of integration paths in \eqref{M-int3} (see \eqref{signs}) implies that, as $x\to +\infty$,
\begin{equation}\label{lims-p}
M_+(x,0,z)  = \begin{pmatrix}
1 &  r_+(z){E}(x,z) & 0\\
0 & 1 & 0\\
0 & 0 & 1
\end{pmatrix} 
+\osmall(1)
\end{equation}
and 
\begin{equation}\label{lims-m}
M_-(x,0,z) = \begin{pmatrix}
1 & 0 & 0\\
r_-(z){E}^{-1}(x,z) & 1 & 0\\
0 & 0 & 1
\end{pmatrix}
+\osmall(1)
\end{equation}
with some $r_\pm(z)$ independent of $x$, and where 
\[
E(x,z) = \ee^{Q_{11}(x,0,z)-Q_{22}(x,0,z)} = \ee^{y(x,0)(\lambda_1(z)-\lambda_2(z))} = \ee^{y(x,0)z(\omega-\omega^2)}.
\]
The symmetry (S1) from Proposition~\ref{p-sym} implies that $r_-(z)=\overline{r_+(z)}$. On the other hand, $M_+$ and $M_-$ are bounded for all $x$ (particularly, as $x\to -\infty$).

Letting $x\to+\infty$ in the r.h.s.\ of 
\[
S_0(z)= \ee^{-Q(x,0,z)}M_-^{-1}(x,0,z) M_+(x,0,z)
\ee^{Q(x,0,z)}
\]
we get that $(S_0)_{31}(z)=(S_0)_{32}(z)\equiv 0$ and $(S_0)_{11}(z)=(S_0)_{33}(z)\equiv 1$. On the other hand, letting $x\to-\infty$ yields $(S_0)_{13}(z)=(S_0)_{23}(z)\equiv 0$. Thus
\begin{equation}\label{S-SS-1}
S_0(z)= 
\begin{pmatrix}
1 & 0 & 0\\
-r(z) & 1 & 0\\
0 & 0 & 1
\end{pmatrix}
\begin{pmatrix}
1 & \overline{r(z)} & 0\\
 0 & 1 & 0\\
0 & 0 & 1
\end{pmatrix},
\end{equation}
where $r(z)\coloneqq r_-(z)=\overline{r_+(z)}$. 

For $z\in l_4=\D{R}_-$, the structure of $S_0(z)$ is similar to \eqref{S-SS-1} whereas for $z\in l_{\nu}$ with $\nu\neq 1,4$, the construction of $S_0(z)$  follows from the symmetries of Proposition~\ref{p-sym}: $S_0(\omega z) = C S_0(z) C^{-1}$.

Thus, as in the case of the Degasperis--Procesi equation,  the jump matrix on the whole contour is determined by a scalar function --- the reflection coefficient $r(k)$ given for $k\in\D{R}$, and has only a $2\times 2$ nontrivial block (different blocks for different $l_{\nu}$'s, according to Proposition~\ref{p-sym}). In turn, as it follows from \eqref{lims-p} and \eqref{M-int3} for $t=0$, the reflection coefficient is determined by the initial condition $u_0(x)$ via the solution of \eqref{M-int3}, where $q(x,0)\equiv (-u_{xx}(x,0)+1)^{1/3}$ in the construction of $ U(x,0,z)$, see \eqref{U}, is replaced by $q_0(x)\equiv (-u_{0xx}(x)+1)^{1/3}$.

\begin{prop}
$r(z)=\ord(z^3)$ as $z\to 0$.
\end{prop}

\begin{proof}
This follows from \eqref{scat}, \eqref{S-SS-1}, and from the fact that the terms 
up to $z^2$ of the expansions of $M_+$ and $M_-$ as $z\to 0$ coincide, see \eqref{M-as}.
\end{proof}

\begin{rem}
The dependence of the matrix $\ee^{Q}S_0\ee^{-Q}$, relating $M_+$ and $M_-$ in \eqref{scat}, on the parameters $(x,t)$ justifies the use of the variable $y=y(x,t)$ in  \eqref{y}. Indeed, introducing
\[
\hat M(y,t,z)\coloneqq M(x(y,t),t,z),
\]
\eqref{scat} can be written in terms of the parameters $(y,t)$ as 
\begin{equation}\label{RH-y}
\hat M_+(y,t,z)=\hat M_-(y,t,z)S(y,t,z),
\end{equation}
where the jump matrix
\begin{equation}\label{S-yt}
S(y,t,z)=\ee^{y\Lambda(z)+t\Lambda^{-1}(z)}S_0(z)\ee^{-y\Lambda(z)-t\Lambda^{-1}(z)}
\end{equation}
is determined in terms of the initial condition $u(x,0)=u_0(x)$ and depends explicitly on the parameters $(y,t)$.
\end{rem}
 
\subsubsection{Normalization condition}
\begin{equation} \label{RH-norm}
\hat M(y,t,z)=I+\ord(1/z)\ \text{ as }\ z\to\infty.
\end{equation}
See (iii) in Proposition \ref{p1}.

\subsubsection{Residue conditions}
 
If $M(x,t,z)$ (or $\hat M(y,t,z)$) has poles in $\D{C}\setminus\Sigma$, the formulation of the RH problem has to be completed by residue conditions at these poles. The following statement holds true.
 
\begin{prop}[see \cite{BC}]
Generically, there are at most a finite number of poles lying in $\D{C}\setminus\Sigma$, each of them being simple, with residue conditions of a special matrix form: distinct columns of $\hat M$ have distinct poles, and if $z_n$ is a pole of the $l$-th column, then for some $j\neq l$
\begin{equation}\label{res}
\Res_{z=z_n}\hat M^{(l)}(y,t,z)=\hat M^{(j)}(y,t,z_n)v_n^{jl}\ee^{y(\lambda_j(z)-\lambda_l(z))+t (\lambda_j^{-1}(z)-\lambda_l^{-1}(z))}
\end{equation}
with some scalar constants $v_n^{jl}$.
\end{prop}

These constants, similarly to the jump matrix $S_0(z)$, are determined by the initial condition $u_0(x)$ via the solution of \eqref{M-int3} considered at $t=0$; thus conditions \eqref{res} have to be consistent with the symmetries of Proposition~\ref{p-sym}.

Now we observe that relations \eqref{RH-y}, \eqref{RH-norm}, \eqref{res} can be viewed as conditions defining a Riemann--Hilbert problem: 

\begin{matrixrhp*}
Given the scattering data $\{S_0(z), z\in\Sigma;\{v_n^{jl}\}_{n=1}^N\}$ (which are determined by $u_0(x)$), find, for each $y\in\D{R}$ and $t\geq 0$, a piecewise (relative to $\Sigma$) meromorphic $3\times 3$ matrix-valued function $\hat M(y,t,z)$ satisfying the jump condition \eqref{RH-y}, the normalization condition \eqref{RH-norm}, and the residue conditions \eqref{res}.
\end{matrixrhp*}

Notice that since $\det S_0\equiv 1$, the solution of this problem (which exists by construction) is unique, by Liouville's theorem.

Summarizing, we arrive at the following

\begin{thm}[representation of $u$ in terms of $\hat M$] \label{thm-1}
Let $u(x,t)$ be the solution of the Cauchy problem \eqref{icc}. Then $u(x,t)$ can be expressed in terms of the solution $\hat M(y,t,z)$ of the matrix Riemann--Hilbert problem formulated above, with data (jump matrix and residue conditions) determined by $u_0(x)$. This expression has a parametric form:
\[
u(x,t)=\hat u(y(x,t),t),
\]
where
\begin{subequations}   \label{u-matr}
\begin{align}\label{x-y-matr}
&x(y,t) = y + \lim_{z\to 0}\left(\frac{\sum_{j=1}^{3}\hat M_{j3}(y,t,z)}{\sum_{j=1}^{3}\hat M_{j3}(y,t,0)}-1\right)\frac{1}{z}, \\ 
\label{u-y-matr}
&\hat u(y,t) = \frac{\partial x(y,t)}{\partial t}.
\end{align}
\end{subequations}
\end{thm}
  
\begin{proof}
In view of \eqref{M-as}, the r.h.s.\ of \eqref{x-y-matr} gives $y+N$, which, in view of \eqref{N}, gives $x(y,t)$. In turn, the form $q_t = -(uq)_x$ of the sDP equation \eqref{ic} and the definition \eqref{y} of $y(x,t)$ imply that $\frac{\partial y}{\partial t}(x,t)=-\int_x^{\infty} q_t \dd \xi = -uq$, which, in view of the identity $0=\frac{\partial y}{\partial x}\frac{\partial x}{\partial t}+\frac{\partial y}{\partial t}$, gives \eqref{u-y-matr}.
\end{proof}

Theorem~\ref{thm-1} gives a \emph{representation} of the solution of the Cauchy problem \eqref{icc} \emph{under assumption of existence of a solution to this problem}. On the other hand, the construction of the RH problem above allows \emph{solving} the Cauchy problem \eqref{icc} \emph{under assumption of existence of a solution to the RH problem satisfying a certain structural condition}.

\begin{thm}[structured matrix RH problem]\label{thm-2}
Let $S_0(k)$ and $\{z_n,v_n^{jl}\}$ be respectively the scattering matrix and the residue parameters determined by $u_0(x)$ via the solution of \eqref{M-int3} at $t=0$. Let the Riemann--Hilbert problem specified by the jump condition \eqref{RH-y}, the residue conditions \eqref{res}, and the normalization condition \eqref{RH-norm} have a unique solution $\hat M(y,t,z)$ satisfying the structural condition at $z=0$:
\begin{equation}\label{M-z0}
\hat M(y,t,0)=
\begin{pmatrix}
\alpha&\beta&\bar\beta\\
\bar\beta&\alpha&\beta\\
\beta&\bar\beta&\alpha
\end{pmatrix}
\end{equation}
where $\alpha(y,t)$ and $\beta(y,t)$ have the form
\begin{equation}\label{al-be-M}
\alpha=\bar{\alpha}=\frac{1}{3}\left(\hat q+1+\frac{1}{\hat q}\right),\quad\beta=\frac{1}{3}\left(\hat q+\omega^2+\frac{\omega}{\hat q}\right)
\end{equation}
with some $\hat q(y,t)>0$ such that $\hat q\to 1$ as $y\to\infty$. Introduce $x=x(y,t)$ and $u=u(x,t)$ as follows:
\begin{enumerate}[$\bullet$]
\item 
$\dfrac{\partial x}{\partial y}=\hat q^{-1}(y,t)$ and $x-y\to 0$ as $y\to +\infty$,
\item
$u(x,t)\coloneqq\dfrac{\partial x}{\partial t}(y,t)\Big\vert_{y=y(x,t)}$,
\item
$q(x,t)\coloneqq q(y(x,t),t)$.
\end{enumerate}
Then:
\begin{enumerate}[\em(i)]
\item 
$\hat q(y,t)$ satisfies the differential equation 
\begin{equation}\label{sDP-y}
\left(\frac{\hat q_y}{\hat q}\right)_t = \hat q^2 - \frac{1}{\hat q}.
\end{equation}
\item
$u(x,t)$ and $q(x,t)$ satisfy the system of equations 
\begin{equation}\label{sDP-sys}
q_t=-(uq)_x,\qquad u_{xx}=1-q^3,
\end{equation}
which is equivalent to \eqref{ic}.   
\item
$u(x,t)$ satisfies the initial condition \eqref{ini-cond}.    
\end{enumerate}
\end{thm}

\begin{rem}
The fact that a solution $\hat M$ of the RH problem satisfies \eqref{M-z0} actually follows from the symmetries of $\hat M$ (see Proposition~\ref{p-sym}), which, in turn, follow from the respective symmetries of $S_0(k)$ and $\{z_n,v_n^{jl}\}$, and the uniqueness of the solution of the RH problem. Thus, it is the representation of $\alpha$ and $\beta$ as in \eqref{al-be-M}, which is a necessary condition in view of \eqref{G}-\eqref{M-as}, that constitutes an \emph{additional condition} that must be imposed on the solution of the RH problem. This situation is different from that in the case of classical integrable equations like the Korteweg--de Vries equation or the nonlinear Schr\"odinger equation, where no conditions on the solution of the respective RH problems are to be imposed. On the other hand, it is typical for the so-called ``peakon equations'', like the CH and DP equations \cites{BmS08c, BmS13}.
\end{rem}

\begin{proof}[Proof of Theorem~\ref{thm-2}]
The first step in the proof that the solution of the RH problem gives rise to a solution of the nonlinear equation in question follows the standard scheme, see, e.g.\ \cite{FT}. Defining $\Psi\coloneqq\hat M\ee^{y\Lambda+t\Lambda^{-1}}$ and introducing the derivatives $\Psi_y$ and $\Psi_t$, one obtains that:
\begin{enumerate}[a)]
\item
$
\Psi_y \Psi^{-1}=z \hat \Lambda + W(y,t) + \ord(1/z)$ \text{as}\ $z\to \infty$,
where $W=[M_1,\hat \Lambda]$ and $M_1$ comes from the expansion 
$\hat M=I + \frac{M_1}{z}+ \ord(1/z^2)$ as $z\to \infty$,
whereas  $\Psi_y \Psi^{-1}$ has no jumps and is bounded in $z\in \mathbb C$.
\item
$\Psi_t \Psi^{-1}= \frac{A(y,t)}{z} + \ord(1)$ \text{as}\ $z\to 0$, where $A=G_0\hat \Lambda^{-1}G_0^{-1}$ and $G_0$ comes from the expansion $\hat M=G_0 + \ord(z)$ as $z\to 0$; $\Psi_t \Psi^{-1} = \ord(1/z)$ as  $ z\to \infty$; and $\Psi_y \Psi^{-1}$ has no jumps and is bounded in $z\in \mathbb C\setminus {0}$.
\end{enumerate}
Then by Liouville's theorem we conclude that $\Psi$ satisfies the system of equations
\begin{subequations}   \label{Lax-y}
\begin{align}\label{Lax-y-y}
&\Psi_y = z\hat \Lambda \Psi + W(y,t) \Psi, \\
\label{Lax-y-t}
& \Psi_t = \frac{A(y,t)}{z}\Psi.
\end{align}
\end{subequations}
Moreover, the compatibility condition $\Psi_{xt}=\Psi_{tx}$ for \eqref{Lax-y} 
leads to the system of equations for $W$ and $A$:
\begin{subequations}   \label{Lax-W-A}
\begin{align}
&A_y = [W,A], \\
& W_t = [A,\hat \Lambda].
\end{align}
\end{subequations}
Now we notice that 
\begin{enumerate}[(i)]
\item
taking into account the symmetries of Proposition~\ref{p-sym} provides  $\hat M(y,t,0)$ with the structure \eqref{M-z0} with $\alpha$ and $\beta$ satisfying the equation $\alpha^3+\beta^3+\bar\beta^3  - 3\alpha\abs{\beta}^2=0$,
\item
it is the additional condition \eqref{al-be-M} that reduces \eqref{Lax-W-A} to a single equation \eqref{sDP-y} for a single function, namely, for $\hat q$.
\end{enumerate}
Further, defining $\hat u(y,t)\coloneqq\frac{\partial x}{\partial t}(y,t)$, it follows that 
\begin{equation} \label{h-u-y}
\frac{\partial\hat u}{\partial y}(y,t)=\left(\frac{1}{\hat q}\right)_t=-\frac{\hat q_t}{\hat q^2}.
\end{equation}
Substituting this into \eqref{sDP-y} yields
\begin{equation} \label{h-q-3}
\hat q^3 - 1 = -\left(\hat u_y \hat q\right)_y\hat q,
\end{equation}
or, in the $(x,t)$ variables, 
\begin{equation} \label{q-3}
q^3 - 1 = -u_{xx}
\end{equation}
where $q(x,t)=\hat q(y(x,t),t)$ and $u(x,t)=\hat u(y(x,t),t)$. Further, taking into account that if $\hat f(y,t)=f(x,t)$, then $\frac{\partial \hat f}{\partial t} = \frac{\partial f}{\partial x}\frac{\partial x}{\partial t} + \frac{\partial f}{\partial t}=\frac{\partial f}{\partial x}u + \frac{\partial f}{\partial t}$, \eqref{sDP-y} and \eqref{q-3} give 
\begin{equation} \label{dif1}
\left(\frac{q_t+u q_x}{q}\right)_x = -u_{xx}.
\end{equation}
Provided $u_x\to 0$ as $x\to +\infty$, integrating \eqref{dif1} finally gives $q_t+u q_x = -q u_x$, or $q_t=-(uq)_x$. 

In order to verify the initial condition, one observes that for $t=0$, the RH problem reduces to that associated with $u_0(x)$, which yields $u(x,0)=u_0(x)$, owing to the uniqueness of the solution of the RH problem.
\end{proof}

\begin{rem}
Introducing $\phi$ by $\hat q=\ee^{2\phi}$, equation \eqref{sDP-y} reduces
to the Bullough--Dodd--Mikhailov equation \cites{DB77,M79}
\[
\phi_{yt}=\frac{1}{2}\left(\ee^{4\phi}-\ee^{-2\phi}\right).
\]
\end{rem}

The analysis of the soliton solutions, see Section~\ref{sec:loop}, suggests making the conjecture that initial data $u_0(x)$ satisfying the condition $-u_{0xx}+1>0$ for all $x$ give rise to a piecewise \emph{analytic} RH problem, i.e., without residue conditions. On the other hand, \emph{forcing} residue conditions (consistent with the symmetries) leads to a representation of non-classical solutions (singular, non-smooth, multivalued, etc.) provided that the solution of the associated RH problem satisfies the structural condition \eqref{M-z0}, \eqref{al-be-M} at $z=0$.

\begin{thm}\label{thm-3}
Let $r(k)$, $k\in\D{R}$ be a smooth function such that $r(k)=\ord(1/k)$ as $\abs{k}\to \infty$ and $r(0)=0$. Let $\{z_n,v_n^{jl}\}$ be a set which is consistent with the symmetries from Proposition~\ref{p-sym}. Assume that the RH problem constructed from the given data above has a solution, which, being evaluated at $z=0$,  satisfies the structural condition \eqref{M-z0}-\eqref{al-be-M}. 

Then $\hat q$ from \eqref{al-be-M} satisfies equation \eqref{sDP-y} and thus $q(x,t)$ and $u(x,t)$, determined from $\hat q$ as in Theorem~\ref{thm-2}, satisfy \eqref{sDP-sys}. Moreover, if, additionally, $r(k)=\ord(k^3)$ as $k\to 0$, then $u\to 0$ and $u_x\to 0$ as $x\to-\infty$.
\end{thm}

\begin{rem}
The facts that $u\to 0$ and $u_x\to 0$ as $x\to +\infty$ follow from the asymptotic analysis in Section~\ref{sec:as}.
\end{rem}
      
\subsection{Vector RH problem}   \label{ssec:vector rhp}

The structure of $M$ at $z=0$, see \eqref{M-as}, suggests the transition from the $3\times 3$ matrix RH problem to a row vector RH problem, for the $1\times 3$ row vector-valued function $\mu$, which can be viewed as the result of left multiplication by the constant row vector $(1\ \ 1\ \ 1)$: 
\[
\mu(y,t,z)=\begin{pmatrix}1&1&1\end{pmatrix}\hat M(y,t,z). 
\]
Then \eqref{RH-y} reduces to the jump condition for $\mu$ across $\Sigma$:
\begin{equation} \label{RH}
\mu_+(y,t,z)=\mu_-(y,t,k)S(y,t,k),
\end{equation}
whereas the residue conditions and the normalization condition take respectively the forms:
\begin{equation}\label{mu-res}
\Res_{z=z_n}\mu_l(y,t,z)=\mu_j(y,t,z_n)v_n^{jl}\ee^{y(\lambda_j(z)-\lambda_l(z))+t (\lambda_j^{-1}(z)-\lambda_l^{-1}(z))}.
\end{equation}
and 
\begin{equation}  \label{vec-norm}
\mu(y,t,z)=\begin{pmatrix}1&1&1\end{pmatrix}+\ord(1/z)\quad\text{as }z\to\infty.
\end{equation}
Then \eqref{M-as} yields the following behavior of $\mu(\,\cdot\,,\,\cdot\,,z)$ as $z\to 0$:
\begin{subequations}   \label{mu-as}
\begin{align}\label{mu1-as}
&\mu_1=\hat q\left(1+z\omega N+\frac{1}{2}z^2\omega^2N^2\right)+\ord(z^3), \\ 
\label{mu2-as}
&\mu_2=\hat q\left(1+z\omega^2N+\frac{1}{2}z^2\omega N^2\right)+\ord(z^3), \\
\label{mu3-as}
&\mu_3=\hat q\left(1+zN+\frac{1}{2}z^2N^2\right)+\ord(z^3).
\end{align}
\end{subequations}
Similarly to the matrix RH problem, relations \eqref{mu-as} indicate that having known $\mu(y,t,z)$ for all $(y,t)$ and for all $z$ near $z=0$ allows reconstructing the solution $u(x,t)$ of the original initial value problem \eqref{icc} in a parametric form as follows: 
\begin{subequations}   \label{u}
\begin{align}
u(x,t)&=\hat u(y(x,t),t),\notag
\intertext{where}\label{x-y}
x(y,t)&=y+\lim_{z\to 0}\left(\frac{\mu_3(y,t,z)}{\mu_3(y,t,0)}-1\right)\frac{1}{z},\\ 
\label{u-y}
\hat u(y,t)&=\frac{\partial x(y,t)}{\partial t}.
\end{align}
\end{subequations}
Notice also the formula for $\hat q$:
\begin{equation}\label{q-y}
\hat q(y,t) = \mu_3(y,t,0).
\end{equation}

In the vector formulation, the structural condition \eqref{M-z0}-\eqref{al-be-M} is obviously lost. On the other hand, the definitions of $y(x,t)$ and $N(y,t)$, see \eqref{y} and \eqref{N}, yield the following \emph{necessary} condition to be satisfied by the coefficients $N$ and $\hat q$ in the expansion \eqref{mu-as}, for small $z$, of the solution of the vector RH problem:
\begin{equation}\label{N-q}
\frac{\partial N}{\partial y}=\frac{1-\hat q}{\hat q}.
\end{equation}

\section{Loop solitons}   \label{sec:loop}

For all $j,l$ and $z$, $(\lambda^j(z)-\lambda^l(z))^2 = -3\lambda_j(z)\lambda_l(z)$ and thus the exponential in \eqref{mu-res}
can be written as 
\begin{equation}\label{exp}
\ee^{y(\lambda_j-\lambda_l)+t(\lambda_j^{-1}-\lambda_l^{-1})} = \ee^{y\nu+t\frac{3}{\nu}},
\end{equation}
where $\nu=\nu(z)=\lambda_j(z)-\lambda_l(z)$.

When integrating soliton nonlinear equations by a Riemann--Hilbert approach, soliton solutions correspond to trivial jump conditions on the associated contour, which reduces the solution of the RH problem to the solution of a system of algebraic equations generated by the residue conditions. In order to have real-valued solutions, the exponentials in the residue conditions have to be real (see \cites{BmS08a, BmS13}). Particularly, for the  equation \eqref{ic}, this means that $\nu(z)$ for $z$ at the pole locations has to be real-valued. Since $\lambda_1(z)-\lambda_2(z)=z(\omega - \omega^2) = z \ii \sqrt{3}$, $\lambda_1(z)-\lambda_3(z)=z(\omega - 1) = - z \sqrt{3}\,\ee^{-\ii\pi/6}$, and $\lambda_3(z)-\lambda_2(z)=z(1 - \omega^2) = z  \sqrt{3}\,\ee^{\ii\pi/6}$, the poles may lie only on the lines $z\in{\ii \D{R}}\cup{\ee^{\ii\pi/6}\D{R}}\cup{\ee^{-\ii\pi/6}\D{R}}$. More precisely, the poles for $\mu_1$,  $\mu_2$, and $\mu_3$ may lie on ${\ii\D{R}}\cup{\ee^{\ii\pi/6}\D{R}}$, ${\ii \D{R}}\cup{\ee^{-\ii\pi/6}\D{R}}$, and ${\ee^{\ii\pi/6}\D{R}}\cup{\ee^{-\ii\pi/6}\D{R}}$, respectively.

Due to the symmetries (see Proposition~\ref{p-sym}), the simplest case involves $6$ poles, at $z=\rho\,\ee^{\frac{\ii\pi}{6}+\frac{\ii\pi m}{3}}$, $m=0,\dots,5$ for some $\rho>0$, with the residue condition
\[
\Res_{z=-\ii\rho}\mu_1(y,t,z)=\mu_2(y,t,\ii\rho)\gamma
\ee^{-\sqrt{3}\rho y-\frac{\sqrt{3}}{\rho}t}
\]
for some constant $\gamma\equiv |\gamma|\ee^{\ii\phi}$ (other residue conditions are determined by the symmetries). Then the solution of the (vector) RH problem has the form
\begin{equation}\label{mu-sol}
\mu = \left(1+\frac{\alpha}{z+\ii\rho}+\frac{\bar\alpha\omega}{z+\rho\ee^{\frac{\ii\pi}{6}}}\,,1+\frac{\bar\alpha}{z-\ii\rho}+\frac{\alpha\omega^2}{z+\rho\ee^{\frac{-\ii\pi}{6}}}\,,1+\frac{\alpha\omega}{z-\rho\ee^{\frac{\ii\pi}{6}}}+\frac{\bar\alpha\omega^2}{z-\rho\ee^{\frac{-\ii\pi}{6}}}\right),
\end{equation}
where 
\begin{align}\label{al}
\alpha &= 2\sqrt{3}\rho\,\frac{\ee^{\ii\phi}\hat e+\hat e^2}{1-4\cos(\phi-\frac{\pi}{3})\hat e+\hat e^2}
\intertext{with}\label{hat-e}
\hat e (y,t) &= \frac{|\gamma|}{2\sqrt{3}\rho}\,\ee^{-\sqrt{3}\rho y- \frac{\sqrt{3}}{\rho}t}=\ee^{-\sqrt{3}\rho(y + \frac{t}{\rho^2} + y_0)}
\intertext{and} 
y_0&= -\frac{1}{\sqrt{3}\rho}\log\frac{|\gamma|}{2\sqrt{3}\rho}\,.\notag
\end{align}
Notice that: 
\begin{enumerate}[\textbullet]
\item 
the structure of $\mu$ in \eqref{mu-sol} for small $z$ is consistent with all equations in \eqref{mu-as};
\item
it yields real-valued $x(y,t)$ (and thus $\hat u(y,t)$ and $\hat q(y,t)$) for all $\alpha\in {\mathbb C}$:
\begin{subequations}   \label{u-sol}
\begin{align}\label{x-y-sol}
&x(y,t) = y - \frac{2}{\rho} \frac{\Re(\alpha\ee^{\frac{\ii\pi}{3}})}{1+\frac{2}{\rho}\Im\alpha},\\
\label{q-y-sol}
&\hat q(y,t) = 1+\frac{2}{\rho}\Im\alpha.
\end{align}
\end{subequations}
\end{enumerate}
Substituting \eqref{al} into \eqref{u-sol} gives the parametric representation for a candidate for one-soliton solution:
\begin{subequations}   \label{u-sol1}
\begin{align}\label{x-y-sol1}
&x(y,t) = y + N(y,t) = y + \frac{2\sqrt{3}}{\rho}\frac{-2\cos(\phi+\frac{\pi}{3})\hat e +\hat e^2}{1-4\cos(\phi+\frac{\pi}{3})\hat e +\hat e^2},\\
\label{u-y-sol1}
&\hat u(y,t) = \frac{\partial x(y,t)}{\partial t} = \frac{12}{\rho^2}\frac{\hat e(\cos(\phi+\frac{\pi}{3}) - \hat e + \cos(\phi+\frac{\pi}{3}) \hat e^2)}
{(1-4\cos(\phi+\frac{\pi}{3})\hat e +\hat e^2)^2},\\
\label{q-y-sol1}
&\hat q(y,t) = \frac{1-4\cos(\phi+\frac{\pi}{3})\hat e +\hat e^2}{1-4\cos(\phi-\frac{\pi}{3})\hat e +\hat e^2}.
\end{align}
\end{subequations}
Now, applying the condition \eqref{N-q} to \eqref{u-sol1} and introducing $\beta\coloneqq -2\cos(\phi+\frac{\pi}{3})$ and $\hat \beta \coloneqq -2 \cos(\phi-\frac{\pi}{3})$  reduces \eqref{N-q} to the equation
\begin{equation}\label{be-be}
-3\,\frac{\beta + 2\hat e + \beta \hat e^2}{1+2\beta \hat e + \hat e^2} = \beta-\hat \beta.
\end{equation}
Equation \eqref{be-be} becomes an identity (in $(y,t)$) only in two cases: 
\begin{enumerate}[(i)]
\item
$\beta=1$ and $\hat \beta = -2$,
\item
$\beta=-1$ and $\hat \beta = 2$.
\end{enumerate}
But in case (ii) we have 
\begin{subequations}   \label{u-sol-bad}
\begin{align}\label{x-y-sol-bad}
&x(y,t)  = y - \frac{2\sqrt{3}}{\rho}\frac{\hat e}{1-\hat e},\\
\label{u-y-sol-bad}
&\hat u(y,t)  = \frac{6}{\rho^2}\frac{\hat e}
{(1-\hat e)^2},
\end{align}
\end{subequations}
which corresponds to an unbounded $\hat u$ (at $y=-\frac{t}{\rho^2}-y_0$, where $\hat e (y,t)=1$) and a singular dependence $x=x(y,t)$. 

\begin{figure}[ht]
\hfill\begin{minipage}[t]{.49\textwidth}
\includegraphics[scale=1]{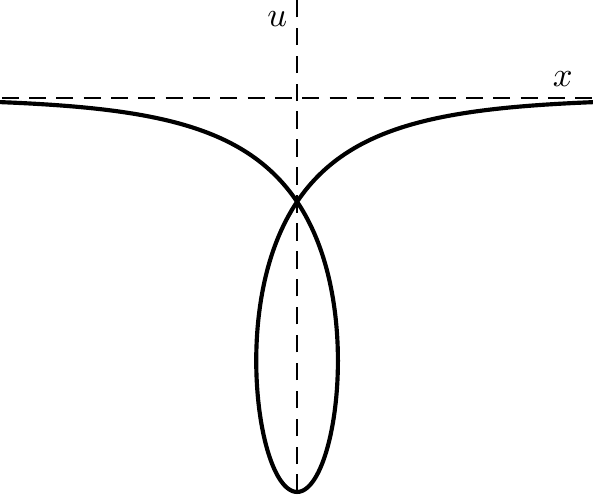}
\caption{The soliton in the $(x,t)$ variables.} 
\label{fig:u-x}
\end{minipage}
\begin{minipage}[t]{.49\textwidth}
\includegraphics[scale=1]{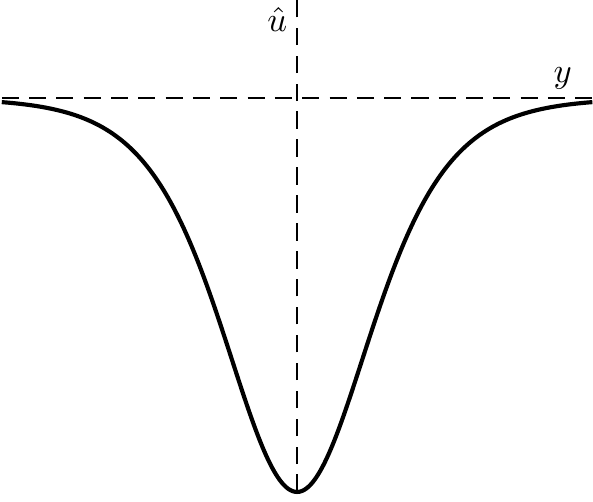}
\caption{The soliton in the $(y,t)$ variables.} 
\label{fig:u-y}
\end{minipage}
\end{figure}

On the other hand, for $\beta=1$ and $\hat \beta = -2$, which, in terms of $\phi$, corresponds to $\phi=\frac{\pi}{3}$, we have a bounded solution described as follows: 
\begin{subequations}   \label{u-sol1-1}
\begin{align}\label{x-y-sol1-1}
&x(y,t)  = y + \frac{2\sqrt{3}}{\rho}\frac{\hat e}{1+\hat e},\\
\label{u-y-sol1-1}
&\hat u(y,t)  = -\frac{6}{\rho^2}\frac{\hat e}
{(1+\hat e)^2},\\
\label{q-y-sol1-1}
&\hat q(y,t) = \frac{1+2\hat e +\hat e^2}{1-4\hat e +\hat e^2},
\end{align}
\end{subequations}
where $\hat e = \hat e(y,t)$ is given by \eqref{hat-e}.

Notice that the representation of the loop solitons in the form \eqref{u-sol1-1} coincide with that in \cite{M06} under an obvious change of notations.

One can directly verify that $q(x,t) = \hat q(y(x,t),t)$ and $u(x,t)=\hat u(y(x,t),t)$ determined by \eqref{u-sol1-1} satisfy the equation \eqref{ic} in the form of the system $q_t = -(uq)_x$, $u_{xx}=1-q^3$.

\begin{figure}[ht]
\includegraphics[scale=0.8]{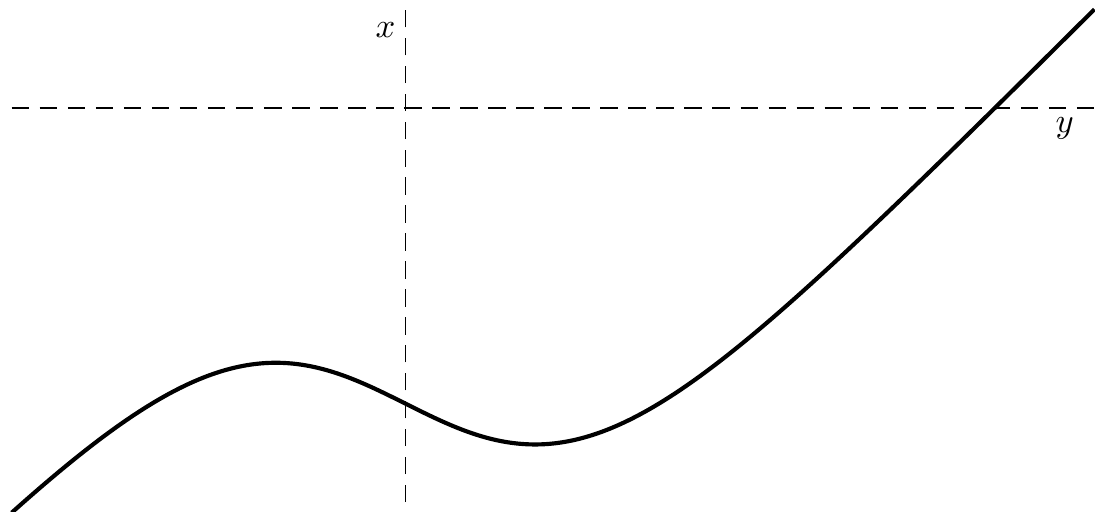}
\caption{The change of variables $y\mapsto x$.} 
\label{fig:x-y}
\end{figure}

\begin{rem}
In the variables $(y,t)$, the soliton solution \eqref{u-y-sol1-1} is a smooth function having a bell shape, see Fig.~\ref{fig:u-y} (corresponding to the change from Euler to Lagrange picture), which is typical for solitons of integrable nonlinear evolution equations. It is the change of variable $y\mapsto x$ \eqref{x-y-sol1-1}, which, for any fixed $t$, is not monotone (see Fig.~\ref{fig:x-y}), that makes the soliton in the original variables $(x,t)$ to be a multivalued function having a loop shape (see Fig.~\ref{fig:u-x}). The physical context of ambiguous (multi-valued) solutions is discussed in \cite{V92}.
\end{rem}

\section{Long time asymptotics} \label{sec:as}

The representation of the solution $u$ to the Cauchy problem \eqref{icc} in terms of an associated RH problem allows applying the nonlinear steepest descent method \cite{DZ93} for obtaining a detailed long time asymptotics of $u$. A key feature of this method is the deformation of the original RH problem according to the ``signature table'' for the phase functions in the jump matrix $S$ in \eqref{RH-y} or \eqref{RH}. Since the structure of the jump matrix is similar, in many aspects, to that in the case of the Degasperis--Procesi equation, the long time analysis shares many common features with that for the DP equation \cite{BmS13}. On the other hand, the short wave limit nature of \eqref{ic} suggests the presence of certain common issues with the case of the short wave Camassa--Holm equation \cite{BSZ11}.

First, consider the structure of the jump matrix $S$ \eqref{S-yt} for $z\in{\D{R}}$. The exponentials in the $(j,l)$ entry of $S$, which has the form \eqref{exp}, can be written as $\ee^{t\left(\zeta \nu +\frac{3}{\nu}\right)}$, where $\zeta = \frac{y}{t}$ and $\nu = \lambda_j(z)-\lambda_l(z) = z(\omega^j - \omega^l)$. Particularly, for the $(1,2)$ entry one has $\nu = \ii\sqrt{3} z$ and thus
\[
\ee^{\ii t \sqrt{3}\left(\zeta z -\frac{1}{z}\right)}\equiv\ee^{-2\ii t\Theta(\zeta,z)},
\]
where $\Theta(\zeta,z)$ has the same form, up to sign and scaling factors, as in the case of the short wave limit of the CH equation \cite{BSZ11}:
\[
\Theta(\zeta, z) = -\frac{\sqrt{3}}{2}\left(\zeta z - \frac{1}{z}\right).
\]
Therefore, the signature table, i.e., the distribution of signs of $\Im \Theta$ in the $z$-plane (near the real axis) is similar to that in \cite{BSZ11}. Particularly, 
\[
\Im \Theta = 0 \iff\Im z\left(\zeta + \frac{1}{|z|^2}\right)=0
\]
and thus, similarly to \cite{BSZ11}, two cases are to be distinguished.
\begin{enumerate}[(i)]
\item 
Case $\zeta\geq 0$. In this case the set $\{z\ | \Im \Theta(\zeta, z)=0\}$ coincides with the real axis $\Im z=0$ and $\pm\Im \Theta> 0$ for $\mp z>0$.
\item
Case $\zeta < 0$. In this case 
\[
\{z\ | \Im \Theta(\zeta, z)=0\} = \{z\ | \Im z=0\} \cup \{z\ |\ |z|=|\zeta|^{-1/2}\}.
\]
\end{enumerate}
Accordingly, the long time behavior of $u(x,t)$ turns out to be qualitatively different for $x>0$ and for $x<0$.

\subsection{Range $\BS{x/t>\varepsilon}$}

In the sector $x/t>\varepsilon$ for any $\varepsilon>0$, in order to move the oscillatory terms, in the jump relation, into regions where they are decaying \cite{DZ93}, the signature table suggests the deformation of the original RH problem according to trigonal factorizations of the jump matrix of type \eqref{S-SS-1}. In what follows, we will assume that $r(z)$ has an analytic extension to a small neighborhood of the real axis. This is, for example, the case if we assume that the solution is exponentially decaying as $|x|\to \infty$. Otherwise one can split $r(z)$ into an analytic part plus a reminder producing a polynomially decaying (in $t$) error term, the decay depending on the rate of decay of the initial condition $u_0(x)$ as $|x|\to \infty$ (see, e.g., \cites{BmKST09, GT09}).

For $z\in {\D{R}}$, writing $S(z)$ in the form 
\begin{equation}\label{S-factor}
S(z)=\begin{pmatrix}
1&0&0\\
-r(z)\ee^{2\ii t\Theta(\zeta,z)}&1&0\\
0&0&1
\end{pmatrix}
\begin{pmatrix}
1&\bar r(z)\ee^{-2\ii t\Theta(\zeta,z)}&0\\
0&1&0\\
0&0&1
\end{pmatrix},\quad z\in\D{R}
\end{equation}
suggests absorbing the triangular factors into the new RH problem, for $\mu^{(1)}$, involving a contour $\Sigma^1=\Sigma^1_+\cup \Sigma^1_-$ near the real axis, see Figure~\ref{sigma-h}, 
\begin{figure}[ht]
\centering\includegraphics[scale=1]{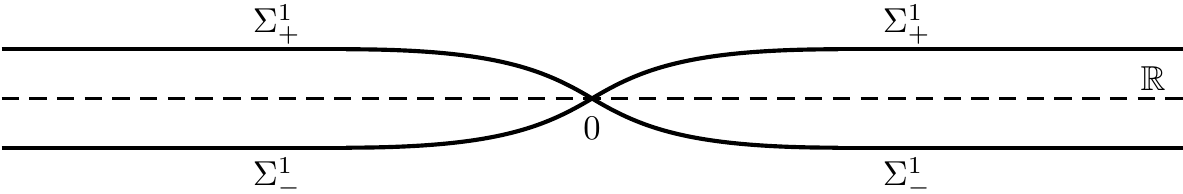}
\caption{Contour $\Sigma^1$ near the real axis} 
\label{sigma-h}
\end{figure}
which replaces the real axis in the original contour:
\[
\mu^{(1)} = \begin{cases}
\mu \begin{pmatrix}
1&0&0\\
-r(z)\ee^{2\ii t\Theta(\zeta,z)}&1&0\\
0&0&1
\end{pmatrix}, & z\ \text{between}\ \D{R} \ \text{and}\ \Sigma^1_-, \\
\mu \begin{pmatrix}
1&-\bar r(z)\ee^{-2\ii t\Theta(\zeta,z)}&0\\
0&1&0\\
0&0&1
\end{pmatrix}, & z\ \text{between}\ \D{R} \ \text{and}\ \Sigma^1_+, \\
\mu, &\text{everywhere else}.
\end{cases}
\]

Similarly, one absorbs the triangular factors associated with the factorizations
of $S$ on the lines $\omega\D{R}$ and $\omega^2\D{R}$, which, due to the symmetries (see 
Proposition~\ref{p-sym}), are as follows:
\begin{alignat*}{2}
S(z)&=\begin{pmatrix}
1&0&-r(\omega^2 z)\ee^{2\ii t\Theta(\zeta,\omega^2 z)}\\
0&1&0\\
0&0&1
\end{pmatrix}
\begin{pmatrix}
1&0&0\\
0&1&0\\
\bar r(\omega^2 z)\ee^{-2\ii t\Theta(\zeta,\omega^2 z)}&0&1
\end{pmatrix},&\quad&z\in \omega\D{R}
\intertext{and}
S(z)&=\begin{pmatrix}
1&0&0\\
0&1&\bar r(\omega z)\ee^{-2\ii t\Theta(\zeta,\omega z)}\\
0&0&1
\end{pmatrix}
\begin{pmatrix}
1&0&0\\
0&1&0\\
0&-r(\omega z)\ee^{2\ii t\Theta(\zeta,\omega z)}&1
\end{pmatrix},&&z\in\omega^2\D{R}.
\end{alignat*}
This reduces the RH problem to a new one, on the contour $\hat\Sigma=\Sigma^1\cup
\omega \Sigma^1 \cup \omega^2 \Sigma^1$, see Figure~\ref{sigma-h2}, 
\begin{figure}[ht]
\includegraphics[scale=1]{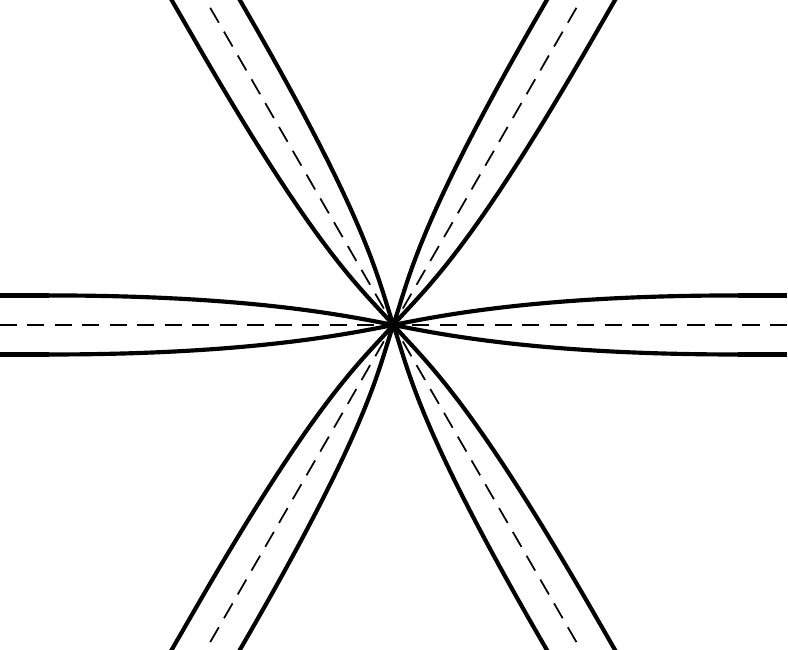}
\caption{Contour $\hat\Sigma$ for $\zeta>0$.} 
\label{sigma-h2}
\end{figure}
whose jump matrix is exponentially decaying (in $t$) to the identity matrix. Assuming there are no residue conditions, the solution to this problem decay fast to $I$ and consequently $\hat u(y,t)$ decay fast to $0$ (while $y$ approaches fast $x$), uniformly in any sector $x/t>\varepsilon$ with $\varepsilon>0$ (see \cite{BSZ11}), and thus $u(x,t)=\ord(t^{-n})$ with $n>1$ depending on the decay of $u_0(x)$ through the decay of the non-analytic reminder in the analytic approximation of the reflection coefficient \cite{GT09}.

\subsection{Range $\BS{x/t<-\varepsilon}$}

In  the sector $\zeta<-\varepsilon$ for any $\varepsilon>0$, the signature table
dictates the use of two types of factorizations of the jump matrix. 
Again consider first the real axis. 
\begin{enumerate}[\textbullet]
\item 
For $z\in (-\varkappa,\varkappa)$, where $\varkappa\equiv\varkappa(\zeta)=1/\sqrt{\abs{\zeta}}$, we use again the factorization \eqref{S-factor}.
\item
For $z\in (-\infty,-\varkappa)\cup (\varkappa, \infty)$, the appropriate factorization is as follows:
\begin{equation}\label{S-factor-1}
S(z)=\begin{pmatrix}
1&\frac{\bar r(z) \ee^{-2\ii t\Theta(\zeta,z)}}{1-\abs{r(z)}^2}&0\\
0&1&0\\
0&0&1
\end{pmatrix}
\begin{pmatrix}
\frac{1}{1-\abs{r(z)}^2} & 0 & 0\\
0 & 1-\abs{r(z)}^2 & 0\\
0&0&1
\end{pmatrix}
\begin{pmatrix}
1&0&0\\
-\frac{r(z)\ee^{-2\ii t\Theta(\zeta,z)}}{1-\abs{r(z)}^2}&1&0\\
0&0&1
\end{pmatrix}.
\end{equation}
Taking into account the symmetries, the factorization \eqref{S-factor-1} suggests introducing the diagonal factor $\tilde \delta(z) \equiv \tilde \delta(\zeta,z)$ (see  \cite{BmS13}):
\begin{equation}   \label{tilde.delta}
\tilde\delta(z)=
\begin{pmatrix}
\delta(z)\delta^{-1}(\omega^2z)&0&0\\
0&\delta^{-1}(z)\delta(\omega z)&0\\
0&0&\delta(\omega^2z)\delta^{-1}(\omega z)
\end{pmatrix}
\end{equation}
where
\begin{equation}\label{delta}
\delta(z)\equiv\delta(\zeta,z)=\exp\left\lbrace\frac{1}{2\ii\pi}\left(\int_{-\infty}^{-\varkappa(\zeta)}+
\int_{\varkappa(\zeta)}^{\infty}\right)\frac{-\log(1-\abs{r(s)}^2)}{s-z}\,\dd s\right\rbrace.
\end{equation}
\end{enumerate}
Then, introducing $\mu^{(1)}\coloneqq\mu\tilde\delta^{-1}(z)$, the jump conditions for $\mu^{(1)}$ are
\[
\mu_+^{(1)}=\mu_-^{(1)}S_1(\zeta,t,z)
\]
with $S_1$ possessing appropriate triangular factorizations with non-diagonal terms decaying to $0$ upon deforming the contour $\Sigma$ to a contour $\hat\Sigma$ close to $\Sigma$, with self-intersection points at $0$, $\pm\varkappa$, $\pm\omega\varkappa$, $\pm\omega^2\varkappa$, see Figure~\ref{sigma-h3}. 
\begin{figure}[ht]
\centering\includegraphics[scale=1]{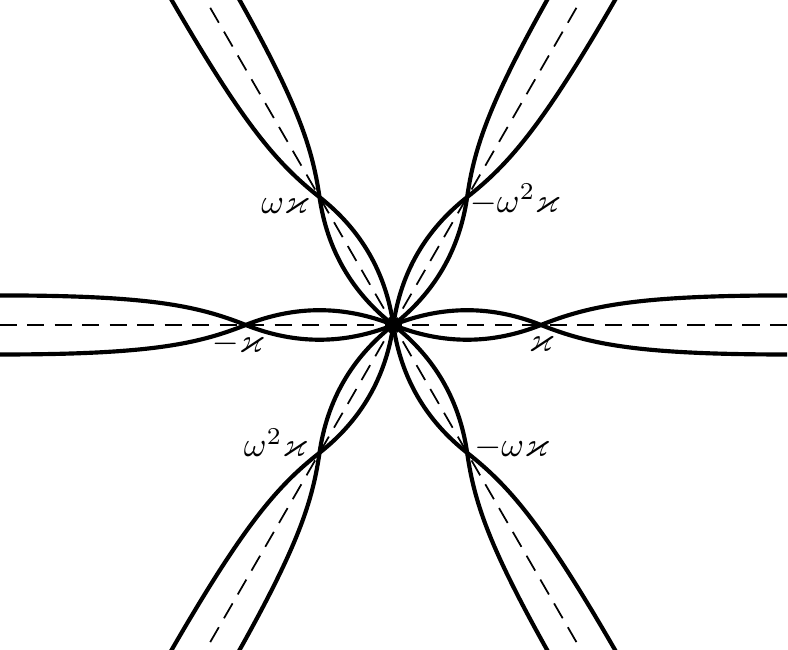}
\caption{Contour $\hat\Sigma$ for $\zeta<0$.} 
\label{sigma-h3}
\end{figure}

\noindent
For instance, for $z\in\D{R}$ (before the deformation) one has
\begin{subequations} \label{eq:jump}
\begin{align}
S_1(\zeta,t,z)&=
\begin{pmatrix}
1&\frac{\bar r(z)}{1-\abs{r(z)}^2}\,\frac{\delta_-^2(z)}{\delta(\omega^2z)\delta(\omega z)}\,\ee^{-2\ii t\Theta(\zeta,z)}&0\\
0&1&0\\
0&0&1
\end{pmatrix}\!\!\!
\begin{pmatrix}
1&0&0\\
\frac{r(z)}{1-\abs{r(z)}^2}\,\frac{\delta(\omega^2z)\delta(\omega z)}{\delta_+^2(z)}\,\ee^{2\ii t\Theta(\zeta, z)}&1&0\\
0&0&1
\end{pmatrix}\notag\\
&\qquad\text{for }z\in(-\infty,-\varkappa)\cup(\varkappa,\infty)
\shortintertext{and}
S_1(\zeta,t,z)&=
\begin{pmatrix}
1&0&0\\
-r(z)\frac{\delta(\omega^2z)\delta(\omega z}{\delta^2(z)}\,\ee^{2\ii t\Theta(\zeta, z)}&1&0\\
0&0&1
\end{pmatrix}\!\!\!
\begin{pmatrix}
1&-\bar r(z)\frac{\delta^2(z)}{\delta(\omega^2z)\delta(\omega z)}\,
\ee^{-2\ii t\Theta(\zeta,z)}&0\\
0&1&0\\
0&0&1
\end{pmatrix}\notag\\
&\qquad\text{for }z\in(-\varkappa,\varkappa).
\end{align}
\end{subequations}
Similarly for $z\in\omega\,\D{R}$ and $z\in\omega^2\,\D{R}$.

Then, absorbing the triangular factors from \eqref{eq:jump} (and the analogous factors from the factorizations on $\omega\,\D{R}$ and $\omega^2\,\D{R}$), the main contribution to the long time asymptotics comes from the sum of (separate) contributions of each small cross (after the contour deformation) centered at $\pm\varkappa$, $\pm\omega\varkappa$, and $\pm\omega^2\varkappa$ (see \cites{DZ93,BmKST09,BmS13}):
\begin{equation}   \label{eq:mu-as}
\mu(z)=\begin{pmatrix}1&1&1\end{pmatrix}\biggl(\sum_{j=1}^{6}M^{(j)}(z)-5\cdot I\biggr)
\tilde\delta(k)+\ord\left(t^{-\alpha}\right)\quad\text{as }t\to\infty,
\end{equation}
with some $\alpha>1/2$, where $M^{(j)}(z)$ are the solutions of the $3\times 3$ matrix RH problems on the crosses centered at $\pm\varkappa$, $\pm\omega\varkappa$, and $\pm\omega^2\varkappa$.

Particularly, for $M^{(1)}(z)$ associated with the cross centered at $\varkappa$, the factors in the jump matrix \eqref{eq:jump} can be approximated, as $t\to\infty$, as follows (see \cites{BmS08b,BmKST09, BmS13}):
\begin{equation}   \label{approx.factors}
\delta^{2}(z)\,\delta^{-1}(\omega^2z)\,\delta^{-1}(\omega z)\,\ee^{-2\ii t\Theta(\zeta,z)}
\simeq\delta_{\ast}^{2}\cdot(-\hat z)^{2\ii h}\ee^{-\ii\hat z^2/2},
\end{equation}
where 
\begin{equation}   \label{h}
h=h(\varkappa)=-\frac{1}{2\pi}\log(1-\abs{r(\varkappa)}^2),
\end{equation}
and where the scaled spectral parameter $\hat z$ is defined by
\[
\hat z=\sqrt{c t}\bigl(z-\varkappa\bigr)
\]
with
\begin{equation}   \label{c}
c=c(\varkappa)=\frac{2\sqrt{3}}{\varkappa^3}\,.
\end{equation}
Moreover, the constant (w.r.t.~$\hat z$) factor $\delta_{\ast}$ is as follows:
\begin{subequations}\label{param}
\begin{align}
\delta_{\ast}&=\ee^{\frac{-\ii\sqrt{3}}{\varkappa}t}\left(\frac{8\sqrt{3}}{\varkappa}t\right)^{-\ii h/2}
\ee^{-\ii\chi_0},
\shortintertext{where}
\chi_0&=\frac{1}{2\pi}\left(\int_{-\infty}^{-\varkappa}+\int_{\varkappa}^{\infty}\right)
\log\abs{\varkappa-s}\dd\log(1-\abs{r(s)}^2)\notag\\
&\quad+\frac{1}{4\pi}\left(\int_{-\infty}^{-\varkappa}+\int_{\varkappa}^{\infty}\right)
\log(1-\abs{r(s)}^2)\,\frac{2s+\varkappa}{s^2+s\varkappa + \varkappa^2}\,\dd s.
\end{align}
\end{subequations}
Similarly, for $z$ near $-\varkappa$, one has
\begin{equation}   \label{approx.factors-1}
\delta^{2}(z)\,\delta^{-1}(\omega^2z)\,\delta^{-1}(\omega z)\,\ee^{-2\ii t\Theta(\zeta,z)}
\simeq\bar\delta_{\ast}^{2}\cdot(\hat z)^{-2\ii h}\ee^{\ii\hat z^2/2},
\end{equation}
where now 
\[
\hat z=\sqrt{c t}(z+\varkappa).
\]
Conjugating out the constant factors in \eqref{approx.factors} and \eqref{approx.factors-1}, the resulting problems on the crosses (in the $\hat z$ plane) become RH problems whose solutions are given in terms of parabolic cylinder functions \cites{DZ93,BmS08b,BmKST09, BmS13}. 
Particularly, $M^{(1)}(z)\approx\Delta_{\one}\hat M^{(1)}\Delta_{\one}^{-1}$ with $\Delta_{\one}=\diag\accol{\delta_{\ast},\delta_{\ast}^{-1},1}$, where the large-$\hat z$ behavior of 
$\hat M^{(1)}(\hat z)$ is given by 
\begin{align*}
\hat M^{(1)}(\hat z)&=I+\frac{\hat M_1}{\hat z}+\ord(\hat z^{-2}),
\shortintertext{where}
\hat M_1&=
\begin{pmatrix}
0&\ii\bar\beta&0\\
-\ii\beta&0&0\\
0&0&0\\
\end{pmatrix}
\end{align*}
with
\begin{equation}   \label{be}
\beta=\frac{r(\varkappa)\Gamma(-\ii h)h}{\sqrt{2\pi}\,\ee^{\ii\pi/4}\ee^{-\pi h/2}}\,.
\end{equation}
Here $\Gamma$ is the Euler Gamma function.

Recalling the relationship $\hat z=\sqrt{c t}\bigl(z-\varkappa\bigr)$, the evaluation of the main term of $M^{(1)}(z)$ as $t\to\infty$, for $z$ close to $0$,
reduces to the following (see \cite{BSZ11}):
\begin{subequations}  \label{eq:m-onesix}
\begin{align} \label{M1}
M^{(1)}(z)&\simeq\Delta_{\one}\biggl(I+\frac{\hat M_1}{\sqrt{c t}\,
	(z-\varkappa)}\biggr)\Delta_{\one}^{-1}\notag\\
&=I+\frac{\vcheck M_1}{\sqrt{c t}\,(-z+\varkappa)}\,,
\shortintertext{where}  \label{M1check}
\vcheck M_1&=\begin{pmatrix}
0&-\ii\bar\beta\delta_{\ast}^2&0\\
\ii\beta\delta_{\ast}^{-2}&0&0\\
0&0&0
\end{pmatrix}.
\end{align}
Similarly, for $M^{(2)}(z)$ associated with the cross centered at $-\varkappa$, one has
\begin{equation}  \label{M2}
M^{(2)}(z)=I+\frac{\vcheck M_2}{\sqrt{ct}(z+\varkappa)}\,,
\end{equation}
where
\[
\vcheck M_2= \begin{pmatrix}
0&\ii\beta^*\delta_{\ast}^{-2}&0\\
-\ii\bar\beta^*\delta_{\ast}^{2}&0&0\\
0&0&0
\end{pmatrix}
\]
with 
\begin{equation}   \label{be-s}
\beta^*=\frac{r(-\varkappa)\Gamma(-\ii h)h}{\sqrt{2\pi}\,\ee^{\ii\pi/4}\ee^{-\pi h/2}}\,.
\end{equation}

For the crosses $M^{(j)}$, $j=3,\dots,6$ centered at $\pm\omega\varkappa$ and  $\pm\omega^2\varkappa$, 
 one applies the symmetries of Proposition~\ref{p-sym}, which gives the following.
\begin{equation}  \label{M3}
M^{(3)}(z)\simeq I+\frac{1}{\sqrt{ct}(-z\omega^2+\varkappa)}
\begin{pmatrix}
0&0&\ii\beta\delta_{\ast}^{-2}\\
0&0&0\\
-\ii\bar\beta\delta_{\ast}^{2}&0&0
\end{pmatrix},
\end{equation}
\begin{equation}  \label{M4}
M^{(4)}(z)\simeq I+\frac{1}{\sqrt{ct}(z\omega^2+\varkappa)}
\begin{pmatrix}
0&0&-\ii\bar\beta^*\delta_{\ast}^{2}\\
0&0&0\\
\ii\beta^*\delta_{\ast}^{-2}&0&0
\end{pmatrix},
\end{equation}
\begin{equation}  \label{M5}
M^{(5)}(z)\simeq I+\frac{1}{\sqrt{ct}(-z\omega+\varkappa)}
\begin{pmatrix}
0&0&0\\
0&0&-\ii\bar\beta\delta_{\ast}^{2}\\
0&\ii\beta\delta_{\ast}^{-2}&0
\end{pmatrix},
\end{equation}
\begin{equation}  \label{M6}
M^{(6)}(z)\simeq I+\frac{1}{\sqrt{ct}(z\omega+\varkappa)}
\begin{pmatrix}
0&0&0\\
0&0&\ii\beta^*\delta_{\ast}^{-2}\\
0&-\ii\bar\beta^*\delta_{\ast}^{2}&0
\end{pmatrix}.
\end{equation}
\end{subequations}
Collecting all relations \eqref{eq:m-onesix} and substituting into \eqref{eq:mu-as} gives
\begin{equation}  \label{mu3}
\mu_3(z)\simeq \tilde\delta_{33}(z)\left\{1+\frac{A_1}{-z\omega^2+\varkappa}
+\frac{\bar A_1}{-z\omega+\varkappa}
+\frac{A_2}{z\omega+\varkappa}
+\frac{\bar A_2}{z\omega^2+\varkappa}
\right\},
\end{equation}
where 
\begin{subequations}  \label{A}
\begin{align}
A_1&=A_1(\zeta)=\ii\beta\delta_{\ast}^{-2},\\
A_2&=A_2(\zeta)=\ii\beta^*\delta_{\ast}^{-2}.
\end{align}
\end{subequations}
In view of \eqref{u}, in order to evaluate the asymptotics of  $u$,
we use expansions, as $z\to 0$, of the r.h.s. of \eqref{mu3} (see \cite{BSZ11}).
From \eqref{tilde.delta} and \eqref{delta} it follows that 
\begin{equation}  \label{delta-as}
\tilde\delta_{33}(z) = 1 + z \Delta + \ord(z^2),
\end{equation}
where 
\begin{equation}  \label{Delta}
\Delta = \Delta(\varkappa(\zeta)) = \frac{\sqrt{3}}{\pi}\int_{\varkappa}^{\infty}\frac{\log\bigl(1-\abs{r(s)}^2\bigr)}{s^2}\dd s,
\end{equation}
and thus
\begin{equation}  \label{mu3-x}
\frac{\mu_3(z)-\mu_3(0)}{\mu_3(0) z}=\Delta + \frac{2}{\varkappa^2\sqrt{ct}}\left(\Re(A_1\omega^2)-\Re(A_2\omega)\right)+\ord\left(t^{-\alpha}\right).
\end{equation}
Using \eqref{h}, \eqref{c}, \eqref{param}, \eqref{be}, \eqref{be-s} in \eqref{A}
yields the following expressions for the terms in the r.h.s.\ of \eqref{mu3-x}:
\begin{subequations}  \label{A-1}
\begin{align}
\Re(A_1\omega^2)&=\sqrt{h}\cos\left(\frac{2\sqrt{3}}{\varkappa}t+h\log t +\phi_1\right),\\
\Re(A_2\omega)&=\sqrt{h}\cos\left(\frac{2\sqrt{3}}{\varkappa}t+h\log t +\phi_2\right),
\end{align}
\end{subequations}
where 
\begin{equation}  \label{phi}
\phi_j=h\log\frac{8\sqrt{3}}{\varkappa} + \arg r\bigl((-1)^{j-1}\varkappa\bigr) +\arg \Gamma(-\ii h)+2\chi_0 +\frac{\pi}{4} +(-1)^j \frac{2\pi}{3}, \quad j=1,2.
\end{equation}
Substituting \eqref{A-1} with \eqref{phi} into \eqref{mu3-x} gives, in view of \eqref{x-y}, the principal term of the asymptotics for $x(y,t)$:
\begin{equation}  \label{x-as}
x(y,t)=y+\Delta + \frac{\tilde c_1}{\sqrt{t}}\sin\left(c_2 t + c_3\log t+ \tilde c_4\right)+\ord\left(t^{-\alpha}\right),
\end{equation}
where the error term is uniform in the sector $x/t<-\varepsilon$. The coefficients $c_j$ and $\tilde c_j$ are, similarly to $\Delta$, functions of $\varkappa$:
\begin{subequations}  \label{c-j}
\begin{align}
\tilde c_1 & = - 2\left(\frac{4}{3}\right)^{1/4}\sqrt{\frac{h}{\varkappa}}\sin\left(\frac{\arg r(\varkappa)- \arg r(-\varkappa)}{2}-\frac{2\pi}{3}\right), \\
c_2&= \frac{2\sqrt{3}}{\varkappa}, \\
c_3&= h,\\
\tilde c_4 &= h\log\frac{8\sqrt{3}}{\varkappa} + \frac{\arg r(\varkappa)+ \arg r(-\varkappa)}{2}+\arg \Gamma(-\ii h) + 2\chi_0 + \frac{\pi}{4}.
\end{align}
\end{subequations}

Now the asymptotics of $u(x,t)$ can be calculated by differentiating \eqref{x-as} with respect to $t$ (keeping $y$ fixed; see \eqref{u-y}) and taking into account the change of variables $y\mapsto x$, which (asymptotically) results in an additional phase shift. This results in the following

\begin{thm} \label{thm:asymptotics}
Let $u(x,t)$ be the solution of the Cauchy problem \eqref{icc}. Then the behavior of $u$ as $t\to\infty$ is as follows. Let $\varepsilon$ be any small positive number.
\begin{enumerate}[\rm(i)]
\item
In the domain $x/t>\varepsilon$, $u(x,t)$ tends to $0$ with fast decay:
\[
u(x,t)=\ord(t^{-n})\text{ for some }n>1.
\]
\item
In the domain $x/t<-\varepsilon$, $u(x,t)$ exhibits decaying, of order $\ord(t^{-1/2})$, modulated oscillations:
\begin{equation}\label{u-as}
u(x,t)=\frac{c_1}{\sqrt{t}}\,\cos(c_2t+c_3\log t+c_4)+\ord\left(t^{-\alpha}\right)
\end{equation}
with some $\alpha>1/2$. The coefficients $c_j$ are functions of $x/t$ given in terms of the associated spectral function $r(z)$:
\begin{align*}
c_1&=-2^{\frac{3}{2}} 3^{\frac{1}{4}}\sqrt{\frac{h}{\vcheck\varkappa^3}}\sin\left(\frac{\arg r(\vcheck\varkappa)- \arg r(-\vcheck\varkappa)}{2}-\frac{2\pi}{3}\right),\\
c_2&=\frac{2\sqrt{3}}{\vcheck\varkappa},\qquad c_3=h,\\
c_4&=\tilde c_4(\vcheck\varkappa)+\sqrt{3}\vcheck\varkappa\Delta(\vcheck\varkappa)\notag\\
&=h\log\frac{8\sqrt{3}}{\vcheck\varkappa}+\frac{\arg r(\vcheck\varkappa)+\arg r(-\vcheck\varkappa)}{2}+\arg \Gamma(-\ii h)\notag\\
&\quad+\frac{\pi}{4}+\frac{3\vcheck\varkappa}{\pi}\int_{\vcheck\varkappa}^\infty\frac{\log\bigl(1-\abs{r(s)}^2\bigr)}{s^2}\dd s \notag\\
&\quad+\frac{1}{\pi}\left(\int_{-\infty}^{-\vcheck\varkappa}+\int_{\vcheck\varkappa}^{\infty}\right)\log\abs{\vcheck\varkappa-s}\dd\log(1-\abs{r(s)}^2)\notag\\
&\quad+\frac{1}{2\pi}\left(\int_{-\infty}^{-\vcheck\varkappa}+\int_{\vcheck\varkappa}^{\infty}\right)\frac{\log(1-\abs{r(s)}^2)(2s+\vcheck\varkappa)}{s^2+s\vcheck\varkappa + \vcheck\varkappa^2}\,\dd s,
\intertext{with}
h&=h(\vcheck\varkappa)=-\frac{1}{2\pi}\log(1-\abs{r(\vcheck\varkappa)}^2)\text{ and }\vcheck\varkappa=\sqrt{\frac{t}{\abs{x}}}.
\end{align*}
\end{enumerate}
\end{thm}

\begin{rems}
The error terms in Theorem \ref{thm:asymptotics} are uniform in the sectors $x/t>\varepsilon$ and $x/t<-\varepsilon$. The matching of the asymptotics for positive and negative values of $x$ is provided by the fast decay of the amplitude $c_1$ in the sector $x/t<-\varepsilon$ as $t/|x|\to \infty$. Indeed, in this limit, the critical point $\varkappa=\sqrt{t/|x|}$ is growing and thus the factor $h=h(\vcheck\varkappa)$ in $c_1$ is decaying to $0$ as fast as the reflection coefficient $r(\vcheck\varkappa)$ is, the latter depending on the smoothness and decay of the initial condition $u_0(x)$. 

Here one can see an analogy with matching the asymptotics for, e.g., the modified Korteweg-de Vries (mKdV) equation, where a similar behavior of the critical points takes place when $x/t$  is approaching $-\infty$, see \cite{DZ93}.  

In this respect we notice that in the case of the mKdV equation, there also exists a specific transition zone matching, for small $x$, the soliton sector and the sector of modulated oscillations \cite{DZ93} (the latter is similar to the sector $x/t<-\varepsilon$ for the Ostrovsky--Vakhnenko equation). In this zone, the main asymptotic term is expressed in terms of a solution of the Painlev\'e II equation. A similar transition zone exists for the Degasperis--Procesi equation (for small $x/t-3$) and for the Camassa--Holm equation (for small $x/t-2$), see \cite{BIS10}. On the other hand, for short-wave approximations of these equations, i.e., for the short-wave model for the CH equation \cite{BSZ11} and the Ostrovsky--Vakhnenko equation, Painlev\'e zones are not present. The appearance of a Painlev\'e zone in the asymptotics of nonlinear equations is indeed characterized by two factors: 
\begin{enumerate}[i)]
\item
At the corresponding point ($x=0$ for mKdV and KdV, $x/t=2$ for CH, $x/t=3$ for DP), there is a bifurcation in the signature table for the associated RH problem.
\item
The value of the reflection coefficient at the corresponding point $\kappa$ is non-zero, so one can define a nontrivial solution of the Painlev\'e II equation (w.r.t.\ $s$) having the asymptotics $r(\kappa)\Ai(s)$ as $s\to +\infty$, where $\Ai(s)$ is the Airy function.
\end{enumerate}
In the case of short-wave equations, none of these properties is satisfied in the seemingly analogous zones adjacent to the sectors of slow decaying, modulated oscillations.
\end{rems}

\section*{Acknowledgments}
D.Sh.\ would like to express his appreciation of the kind hospitality of the University Paris Diderot, where this research was initiated. The work was supported in part by the grant ``Network of Mathematical Research 2013--2015''.



\begin{bibdiv}
\begin{biblist}
\bib{BC}{article}{
   author={Beals, R.},
   author={Coifman, R. R.},
   title={Scattering and inverse scattering for first order systems},
   journal={Comm. Pure Appl. Math.},
   volume={37},
   date={1984},
   number={1},
   pages={39--90},
}
\bib{BDT87}{article}{
   author={Beals, R.},
   author={Deift, Percy},
   author={Tomei, Carlos},
   title={Inverse scattering for selfadjoint $n$th order differential
   operators on the line},
   conference={
      title={Differential equations and mathematical physics (Birmingham,
      Ala., 1986)},
   },
   book={
      series={Lecture Notes in Math.},
      volume={1285},
      publisher={Springer},
      place={Berlin},
   },
   date={1987},
   pages={26--38},
}
\bib{BIS10}{article}{
   author={Boutet de Monvel, Anne},
   author={Its, Alexander},
   author={Shepelsky, Dmitry},
   title={Painlev\'e-type asymptotics for the Camassa-Holm equation},
   journal={SIAM J. Math. Anal.},
   volume={42},
   date={2010},
   number={4},
   pages={1854--1873},
}
\bib{BmKST09}{article}{
   author={Boutet de Monvel, Anne},
   author={Kostenko, Aleksey},
   author={Shepelsky, Dmitry},
   author={Teschl, Gerald},
   title={Long-time asymptotics for the Camassa-Holm equation},
   journal={SIAM J. Math. Anal.},
   volume={41},
   date={2009},
   number={4},
   pages={1559--1588},
}
\bib{BmS08b}{article}{
   author={Boutet de Monvel, Anne},
   author={Shepelsky, Dmitry},
   title={Long-time asymptotics of the Camassa-Holm equation on the line},
   conference={
      title={Integrable systems and random matrices},
   },
   book={
      series={Contemp. Math.},
      volume={458},
      publisher={Amer. Math. Soc.},
      place={Providence, RI},
   },
   date={2008},
   pages={99--116},
}
\bib{BmS08a}{article}{
   author={Boutet de Monvel, Anne},
   author={Shepelsky, Dmitry},
   title={Riemann-Hilbert problem in the inverse scattering for the
   Camassa-Holm equation on the line},
   conference={
      title={Probability, geometry and integrable systems},
   },
   book={
      series={Math. Sci. Res. Inst. Publ.},
      volume={55},
      publisher={Cambridge Univ. Press},
      place={Cambridge},
   },
   date={2008},
   pages={53--75},
}
\bib{BmS08c}{article}{
   author={Boutet de Monvel, Anne},
   author={Shepelsky, Dmitry},
   title={The Camassa-Holm equation on the half-line: a Riemann-Hilbert
   approach},
   journal={J. Geom. Anal.},
   volume={18},
   date={2008},
   number={2},
   pages={285--323},
}
\bib{BmS13}{article}{
   author={Boutet de Monvel, Anne},
   author={Shepelsky, Dmitry},
   title={A Riemann-Hilbert approach for the Degasperis-Procesi equation},
   journal={Nonlinearity},
   volume={26},
   date={2013},
   number={7},
   pages={2081--2107},
}  
\bib{BSZ11}{article}{
   author={Boutet de Monvel, Anne},
   author={Shepelsky, Dmitry},
   author={Zielinski, Lech},
   title={The short-wave model for the Camassa-Holm equation: a
   Riemann-Hilbert approach},
   journal={Inverse Problems},
   volume={27},
   date={2011},
   number={10},
   pages={105006, 17},
}
\bib{B05}{article}{
   author={Boyd, John P.},
   title={Ostrovsky and Hunter's generic wave equation for weakly dispersive
   waves: matched asymptotic and pseudospectral study of the paraboloidal
   travelling waves (corner and near-corner waves)},
   journal={European J. Appl. Math.},
   volume={16},
   date={2005},
   number={1},
   pages={65--81},
}
\bib{BC02}{article}{
   author={Boyd, John P.},
   author={Chen, Guan-Yu},
   title={Five regimes of the quasi-cnoidal, steadily translating waves of
   the rotation-modified Korteweg-de Vries (``Ostrovsky'') equation},
   journal={Wave Motion},
   volume={35},
   date={2002},
   number={2},
   pages={141--155},
}
\bib{BS13}{article}{
   author={Brunelli, J. C.},
   author={Sakovich, S.},
   title={Hamiltonian structures for the Ostrovsky-Vakhnenko equation},
   journal={Commun. Nonlinear Sci. Numer. Simul.},
   volume={18},
   date={2013},
   number={1},
   pages={56--62},
}
\bib{CDG76}{article}{
   author={Caudrey, P. J.},
   author={Dodd, R. K.},
   author={Gibbon, J. D.},
   title={A new hierarchy of Korteweg-de Vries equations},
   journal={Proc. Roy. Soc. London Ser. A},
   volume={351},
   date={1976},
   number={1666},
   pages={407--422},
}
\bib{C01}{article}{
   author={Constantin, Adrian},
   title={On the scattering problem for the Camassa-Holm equation},
   journal={R. Soc. Lond. Proc. Ser. A Math. Phys. Eng. Sci.},
   volume={457},
   date={2001},
   number={2008},
   pages={953--970},
}
\bib{D13}{article}{
   author={Davidson, Melissa},
   title={Continuity properties of the solution map for the generalized reduced Ostrovsky equation},
   journal={J. Differential Equations},
   volume={252},
   date={2013},
   number={6},
   pages={3797--3815},
}
\bib{DP99}{article}{
   author={Degasperis, A.},
   author={Procesi, M.},
   title={Asymptotic integrability},
   conference={
      title={Symmetry and perturbation theory},
      address={Rome},
      date={1998},
   },
   book={
      publisher={World Sci. Publ., River Edge, NJ},
   },
   date={1999},
   pages={23--37},
}
\bib{DZ93}{article}{
   author={Deift, P.},
   author={Zhou, X.},
   title={A steepest descent method for oscillatory Riemann-Hilbert
   problems. Asymptotics for the MKdV equation},
   journal={Ann. of Math. (2)},
   volume={137},
   date={1993},
   number={2},
   pages={295--368},
}
\bib{DB77}{article}{
   author={Dodd, R. K.},
   author={Bullough, R. K.},
   title={Polynomial conserved densities for the sine-Gordon equations},
   journal={Proc. Roy. Soc. London Ser. A},
   volume={352},
   date={1977},
   number={1671},
   pages={481--503},
}
\bib{FT}{book}{
   author={Faddeev, Ludwig D.},
   author={Takhtajan, Leon A.},
   title={Hamiltonian methods in the theory of solitons},
   series={Classics in Mathematics},
   edition={Reprint of the 1987 English edition},
   note={Translated from the 1986 Russian original by Alexey G. Reyman},
   publisher={Springer},
   place={Berlin},
   date={2007},
   pages={x+592},
}
\bib{FO82}{article}{
   author={Fuchssteiner, Benno},
   author={Oevel, Walter},
   title={The bi-Hamiltonian structure of some nonlinear fifth- and
   seventh-order differential equations and recursion formulas for their
   symmetries and conserved covaria},
   journal={J. Math. Phys.},
   volume={23},
   date={1982},
   number={3},
   pages={358--363},
}
\bib{GT09}{article}{
   author={Grunert, Katrin},
   author={Teschl, Gerald},
   title={Long-time asymptotics for the Korteweg-de Vries equation via
   nonlinear steepest descent},
   journal={Math. Phys. Anal. Geom.},
   volume={12},
   date={2009},
   number={3},
   pages={287--324},
}
\bib{HW03}{article}{
   author={Hone, Andrew N. W.},
   author={Wang, Jing Ping},
   title={Prolongation algebras and Hamiltonian operators for peakon
   equations},
   journal={Inverse Problems},
   volume={19},
   date={2003},
   number={1},
   pages={129--145},
}
\bib{H90}{article}{
   author={Hunter, John K.},
   title={Numerical solutions of some nonlinear dispersive wave equations},
   conference={
      title={Computational solution of nonlinear systems of equations (Fort
      Collins, CO, 1988)},
   },
   book={
      series={Lectures in Appl. Math.},
      volume={26},
      publisher={Amer. Math. Soc.},
      place={Providence, RI},
   },
   date={1990},
   pages={301--316},
}
\bib{GHJ12}{article}{
   author={Grimshaw, R. H. J.},
   author={Helfrich, Karl},
   author={Johnson, E. R.},
   title={The reduced Ostrovsky equation: integrability and breaking},
   journal={Stud. Appl. Math.},
   volume={129},
   date={2012},
   number={4},
   pages={414--436},
}
\bib{KM11}{article}{
   author={Khusnutdinova, K. R.},
   author={Moore, K. R.},
   title={Initial-value problem for coupled Boussinesq equations and a
   hierarchy of Ostrovsky equations},
   journal={Wave Motion},
   volume={48},
   date={2011},
   number={8},
   pages={738--752},
}
\bib{KLM11}{article}{
   author={Kraenkel, Roberto A.},
   author={Leblond, Herv\'e},
   author={Manna, Miguel A.},
   title={An integrable evolution equation for surface waves in deep water},
   date={2011},
   eprint={http://arxiv.org/abs/1101.5773v1},
}
\bib{LM06}{article}{
   author={Linares, Felipe},
   author={Milan{\'e}s, Aniura},
   title={Local and global well-posedness for the Ostrovsky equation},
   journal={J. Differential Equations},
   volume={222},
   date={2006},
   number={2},
   pages={325--340},
}
\bib{M06}{article}{
   author={Matsuno, Yoshimasa},
   title={Cusp and loop soliton solutions of short-wave models for the
   Camassa-Holm and Degasperis-Procesi equations},
   journal={Phys. Lett. A},
   volume={359},
   date={2006},
   number={5},
   pages={451--457},
}
\bib{M79}{article}{
   author={Mikhailov, A. V.},
   title={Integrability of a two-dimensional generalization of the Toda chain},
   journal={Pis'ma Zh. Eksp. Teor. Fiz.},
   volume={30},
   date={1979},
   number={7},
   pages={443--448},
}
\bib{MPV99}{article}{
   author={Morrison, A. J.},
   author={Parkes, E. J.},
   author={Vakhnenko, V. O.},
   title={The $N$ loop soliton solution of the Vakhnenko equation},
   journal={Nonlinearity},
   volume={12},
   date={1999},
   number={5},
   pages={1427--1437},
}
\bib{O78}{article}{
   author={Ostrovsky, L. A.},
   title={Nonlinear internal waves in a rotating ocean},
   journal={Oceanology},
   volume={18},
   date={1978},
   number={2},
   pages={181--191},
}
\bib{P93}{article}{
   author={Parkes, E. J.},
   title={The stability of solutions of Vakhnenko's equation},
   journal={J. Phys. A},
   volume={26},
   date={1993},
   number={22},
   pages={6469--6475},
}
\bib{SK74}{article}{
   author={Sawada, Katuro},
   author={Kotega, Takeyasu},
   title={A method for finding $N$-soliton solutions for the KdV equation and KdV-like equation},
   journal={Prog. Theor. Phys.},
   volume={51},
   date={1974},
   number={5},
   pages={1355--1367},
}
\bib{SSK10}{article}{
   author={Stefanov, Atanas},
   author={Shen, Yannan},
   author={Kevrekidis, P. G.},
   title={Well-posedness and small data scattering for the generalized
   Ostrovsky equation},
   journal={J. Differential Equations},
   volume={249},
   date={2010},
   number={10},
   pages={2600--2617},
}
\bib{S06}{article}{
   author={Stepanyants, Y. A.},
   title={On stationary solutions of the reduced Ostrovsky equation:
   periodic waves, compactons and compound solitons},
   journal={Chaos Solitons Fractals},
   volume={28},
   date={2006},
   number={1},
   pages={193--204},
}
\bib{VL04}{article}{
   author={Varlamov, V.},
   author={Liu, Yue},
   title={Cauchy problem for the Ostrovsky equation},
   journal={Discrete Contin. Dyn. Syst.},
   volume={10},
   date={2004},
   number={3},
   pages={731--753},
}
\bib{V92}{article}{
   author={Vakhnenko, V. O.},
   title={Solitons in a nonlinear model medium},
   journal={J. Phys. A},
   volume={25},
   date={1992},
   number={15},
   pages={4181--4187},
}
\bib{V97}{article}{
   author={Vakhnenko, V. O.},
   title={The existence of loop-like solutions of a model evolution equation},
   journal={Ukr. Journ. Phys.},
   volume={42},
   date={1997},
   number={1},
   pages={104--110},
}
\bib{V99}{article}{
   author={Vakhnenko, V. O.},
   title={High-frequency soliton-like waves in a relaxing medium},
   journal={J. Math. Phys.},
   volume={40},
   date={1999},
   number={4},
   pages={2011--2020},
}
\bib{VP98}{article}{
   author={Vakhnenko, V. O.},
   author={Parkes, E. J.},
   title={The two loop soliton solution of the Vakhnenko equation},
   journal={Nonlinearity},
   volume={11},
   date={1998},
   number={6},
   pages={1457--1464},
}
\bib{VP02}{article}{
   author={Vakhnenko, V. O.},
   author={Parkes, E. J.},
   title={The calculation of multi-soliton solutions of the Vakhnenko
   equation by the inverse scattering method},
   journal={Chaos Solitons Fractals},
   volume={13},
   date={2002},
   number={9},
   pages={1819--1826},
}
\bib{VP12}{article}{
   author={Vakhnenko, V. O.},
   author={Parkes, E. J.},
   title={The singular solutions of a nonlinear evolution equation taking
   continuous part of the spectral data into account in inverse scattering
   method},
   journal={Chaos Solitons Fractals},
   volume={45},
   date={2012},
   number={6},
   pages={846--852},
}
\bib{W10}{article}{
   author={Wazwaz, A.-M.},
   title={$N$-soliton solutions for the Vakhnenko equation and its generalized forms},
   journal={Phys. Scr.},
   volume={82},
   date={2010},
   number={6},
   pages={065006, 7 pages},
}
\end{biblist}
\end{bibdiv}
\end{document}